\documentclass[a4paper,USenglish,thm-restate]{lipics-v2021}

\pdfoutput=1
\hideLIPIcs
\nolinenumbers
\usepackage{cite}
\usepackage{booktabs}
\hypersetup{colorlinks,linkcolor=magenta,citecolor=magenta,urlcolor=blue}

\title{Extending Higher-Order Critical Pairs}

\author{Johannes Niederhauser}{University of Innsbruck, Innsbruck, Austria}
{johannes.niederhauser@uibk.ac.at}{https://orcid.org/0000-0002-8662-6834}{}
\author{Aart Middeldorp}{University of Innsbruck, Innsbruck, Austria}
{aart.middeldorp@uibk.ac.at}{https://orcid.org/0000-0001-7366-8464}{}

\authorrunning{J. Niederhauser and A. Middeldorp} 

\Copyright{Johannes Niederhauser and Aart Middeldorp}



\supplement{\url{https://github.com/niedjoh/hrstk}}

\usepackage{definitions}

\begin{document}

\maketitle

\begin{abstract}
Confluence is a fundamental property of computation which implies
that the nondeterminism of a system does not affect its outcomes. A
key concept of confluence research in rewriting theory is the notion
of critical pairs which captures the potential local sources of
non-confluence. In this paper, we lift the definition of higher-order
critical pairs to a more expressive class of higher-order rewrite systems.
Thereby, we facilitate confluence analysis of rewrite systems which
can express complex one-step computations as well as program
transformations from functional programming.
\end{abstract}

\section{Introduction}
\label{sec:introduction}

Term rewriting \cite{BN98,T03} is a Turing-complete model of
computation which underlies much of functional programming.%
\footnote{For now, we deliberately do not specify what exactly we mean
by ``term'' as the focus lies on the general idea of term rewriting.}
The two distinguishing features which set it apart from functional
programming are that
\begin{enumerate}[(i)]
\item
rules are applied nondeterministically, and
\item
the left-hand sides of rules can match arbitrary terms.
\end{enumerate}
Regarding (i), term rewrite systems (TRSs) represent sets of rules as
opposed to functional programming which tries to match the defining
equations of a function in a fixed order according to the program.  On
top of that, TRSs are not limited to specific evaluation strategies
like call by name or call by value.  Regarding (ii), TRSs can easily
express program transformations like
\[
\m{map}\:f\:(\m{map}\:g\:xs) \R \m{map}\:(f \mathrel{.} g)\:xs
\]
which fuses two applications of the $\m{map}$ function into one.
In addition, functional programs are usually \emph{left-linear}
(variables occur at most once in each left-hand side)
while TRSs do not have this restriction.

Given a
TRS $\xR$, we write $s \RbR t$ ($s \Rab{\xR}{*} t$)
if there is a rewrite step (rewrite sequence) from $s$ to $t$.
From the rule above we can
compute the following rewrite sequence:
\begin{align*}
\m{map}\:(\div 3)\: (\m{map}\:(+ 1)\:(\m{map}\:(\times 2)\:[1,2]))
&\RbR \m{map}\:(\div 3)\: (\m{map}\:((+1) \mathrel{.} (\times 2))\:[1,2])
\\
&\RbR \m{map}\:((\div 3) \mathrel{.} ((+1) \mathrel{.} (\times 2)))\:[1,2]
\end{align*}
If $s \Rab{\xR}{*} t$ and $t$ cannot be rewritten anymore, we say that
$t$ is a \emph{normal form} of $\xR$ and write $s \Rab{\xR}{!} t$.
Research in rewriting often focuses on one of the two central
questions: Given a TRS $\xR$,
\begin{enumerate}[(i)]
\item
are there infinite rewrite sequences $t_1 \RbR t_2 \RbR \cdots\:$?
\item
do $s \Rab{\xR}{!} t$ and $s \Rab{\xR}{!} u$ imply $t = u$?
\end{enumerate}
If (i) is answered negatively, we say that the TRS is \emph{terminating}.
If (ii) is true then the TRS has \emph{unique normal forms}. Given a TRS
$\xR$, consider the relation $\Rab{\xR}{!}$. Whenever
$\xR$ has unique normal forms, $\Rab{\xR}{!}$ represents a
partial function. If $\xR$ is also terminating,
$\Rab{\xR}{!}$ is a total function. The relevance of termination
analysis in term rewriting directly carries over to functional programs.
Unique normal forms are important for functional languages
which support parallelism, concurrency, or other nondeterministic
operations.

\subsection{Critical Pairs}

Instead of establishing that a rewrite system has unique normal forms,
one usually shows a stronger property called \emph{confluence}.
A TRS is confluent if 
diverging rewrite sequences can always be joined. More formally,
in a confluent TRS $\xR$, for every \emph{peak}
$s \Lab{\xR}{*} \cdot \Rab{\xR}{*} t$
we have $s \Rab{\xR}{*} \cdot \Lab{\xR}{*} t$.\footnote{%
Here, $\cdot$ denotes the composition of binary relations.}
In this case, we say that $s$ and $t$ are \emph{joinable}.
A peak is joinable whenever its endpoints are joinable.
For \emph{local confluence}
we only demand joinability of \emph{local peaks}
$s \Lb{\xR} \cdot \Rb{\xR} t$. By Newman's Lemma \cite{N42},
local confluence together with termination implies confluence and therefore
unique normal forms. A fundamental result of
confluence research due to Knuth and Bendix is that local confluence of a
first-order TRS can be reduced to joinability of its
\emph{critical pairs} \cite{KB70,H80}.
Intuitively, critical pairs are the endpoints of the local peaks
which potentially destroy local confluence. These local peaks
are also called \emph{critical peaks} and can be computed from the
rules of a TRS. A
finite and first-order TRS only has finitely many critical pairs,
so the result of Knuth and Bendix
transforms the potentially
infinite task of checking joinability of all possible local peaks
to checking joinability of finitely many critical pairs.
Thus, confluence is decidable for first-order terminating TRSs.
The second basic result about critical pairs states that
\emph{orthogonal} TRSs, i.e., left-linear TRSs without critical pairs,
are confluent.

Nipkow
lifted the ``critical pair approach'' to confluence from first-order
rewriting to rewriting on simply-typed $\lambda$-terms \cite{N91,N93a}.
The corresponding TRSs are called higher-order rewrite systems (HRSs).
In HRSs, terms are viewed modulo the meta-theory of
$\beta\eta$-conversion in the simply-typed $\lambda$-calculus. This
makes it possible to view rules of several (untyped) lambda calculi as
object-level rewrite rules and study term rewriting and
$\lambda$-calculus in a unified theory. Note that HRSs 
are in some sense closer to functional programs than first-order TRSs
as they operate on higher-order terms. However, higher-order matching
modulo $\beta\eta$ is even further away from pattern matching in functional
programs than full first-order matching. In fact, HRSs cannot only
directly model more functional programs than first-order TRSs but also
succinctly express more involved program transformations for these
programs.

The definition of critical pairs always relies on the underlying
unification theory. For first-order terms, this is particularly easy
due to decidability of the problem as well as the existence of most
general unifiers. Higher-order unification of terms in simple type
theory, however, is undecidable and infinitary \cite{H75,SG89,D01}. Thus,
the definition of critical pairs given by Nipkow \cite{N91} is restricted
to a subclass of HRSs called pattern rewrite systems
(PRSs). In PRSs, the left-hand sides of rules are
instances of Miller's patterns \cite{M91}. Like first-order terms,
Miller's patterns enjoy a decidable and unitary unification problem
with most general unifiers. 

\subsection{Contributions}

Even though PRSs can express many interesting
examples including various $\lambda$-calculi (see \exaref{prs}),
there are interesting
higher-order computation rules which are out of scope for PRSs. These
include involved one-step computations (see
Examples~\ref{exa:round} and \ref{exa:sdu})
as well as function synthesis in program transformations of functional
programming (see \exaref{fusion} and \secref{progtrans}).
Hence, the goal of this paper is to
provide a conservative extension of the definition of critical pairs
for PRSs to a class of HRSs which strictly includes PRSs but can also
express the aforementioned higher-order rules. From a practical point
of view, it makes sense to keep matching deterministic; otherwise
there are multiple ways to apply a single rule at the same
position. The deterministic higher-order patterns (DHPs) of
Yokoyama et al.~\cite{YHT04}
capture a large class of simply-typed $\lambda$-terms for
which matching is deterministic while still exhibiting a
matching problem with linear time complexity. In particular, every
pattern in the sense of Miller is also a DHP. Most importantly,
the left-hand sides of interesting higher-order computation rules
like the ones mentioned above can be expressed by DHPs.
Hence, we will
introduce the class of deterministic higher-order pattern rewrite
systems (DPRSs) with DHPs as left-hand sides and lift
the definition of critical pairs from PRSs to DPRSs.
In particular, we will prove that
\begin{itemize}
\item
DPRSs are locally confluent if and only if their critical pairs are
joinable, and
\item
orthogonal DPRSs are confluent.
\end{itemize}

We also provide a prototype implementation of our results.\footnote{%
\url{https://github.com/niedjoh/hrstk}}
The implementation is based on our sound and complete unification procedure
for DHPs which produces minimal complete sets of unifiers~\cite{NM26a}.
In general, minimal complete sets of
unifiers of DHPs and therefore the set of critical pairs of DPRSs may be
infinite which complicates questions on decidability. Therefore, we will
show that these sets stay finite for finite left-linear DPRSs. Note 
that this is only a sufficient condition; neither our local confluence
result nor our prototype implementation are limited to this class.

\subsection{Related Work}

We are not the first ones which try to lift the notion of higher-order
critical pairs beyond Miller's patterns. For second-order systems,
this has indeed been already done by Hamana \cite{H19}, but we
will show in
\exaref{hamana1} that the confluence result given there is actually
unsound. Moreover, Hamana uses a modified version of
functions-as-constructors
higher-order unification (FCU) \cite{LM22}, another extension of
Miller's pattern unification. FCU is decidable and 
produces most general unifiers. However, the restrictions it imposes
are not local to a given term but global to the whole unification
problem. This renders its application in critical pair computation
incomplete as the unification problem may not be defined on certain
combinations of terms which cannot be excluded by locally restricting
left-hand sides of rules. In \exaref{varOverlap} we will see that
this happens frequently due to the usage of lifters introduced
in \secref{lifters}.

\subsection{Outline}

The remainder of this paper is structured as follows: After presenting
the preliminaries in \secref{preliminaries}, DHPs as well as DPRSs are
discussed in \secref{dhp}. We introduce our novel definition of
critical pairs in \secref{cps} and prove the basic confluence
results regarding critical pairs in \secref{conf}. Finally, we
showcase the applicability of DPRSs for program transformations of
functional programs in \secref{progtrans} before concluding in
\secref{conclusion}.

\section{Preliminaries}
\label{sec:preliminaries}

A binary relation $R$ on a set $A$ is \emph{well-founded} if there is
no infinite sequence $a_0 \mathrel{R} a_1 \mathrel{R} \cdots$ where
$a_i \in A$ for $i \in \mathbb{N}$. The composition of two binary
relations $R$ and $S$ is defined as follows:
$a \mathrel{R} \cdot \mathrel{S} b$ if there exists an element $c$
such that $a \mathrel{R} c$ and $c \mathrel{S} b$. For a binary
relation denoted by $\R$ we write $\L$, $\LR$, $\Ra{=}$, $\Ra{+}$ and
$\Ra{*}$ for its inverse, symmetric closure, reflexive closure,
transitive closure and transitive-reflexive closure, respectively.
Given $\R$ we also define the \emph{joinability relation}
${\D} = {\Ra{*} \cdot \La{*}}$. Throughout this text, we will denote a
sequence $\seq{a}$ by $\vec{a}$ where $n \geq 0$ and the corresponding
set by $\SEt{\vec{a}}$. For every $a$, $\vec[0]{a}$ represents the
empty sequence which does not have a denotation. For a binary
relation $R$, $\vec{a} \mathrel{R} \vec{b}$ abbreviates
$a_1 \mathrel{R} b_1, \dots, a_n \mathrel{R} b_n$. Similarly, for a
postfix function $f$, $\vec{a}f$
denotes the pointwise application of $f$ to $\vec{a}$. Given two
sequences $\vec{a}$ and $\vec[m]{b}$, their concatenation is written
as $\vec{a},\vec[m]{b}$.

\subsection{Terms}

In this paper, we consider higher-order rewrite systems (HRSs) as
introduced by Wolfram and Nipkow \cite{W93,MN98} which use simply-typed
$\lambda$-terms in $\beta\eta$-long normal form \cite{C40,BDS13}.
Instead of viewing HRSs as rewrite systems modulo $\beta\eta$, we follow
another tradition in higher-order rewriting where terms stay in
$\beta\eta$-long form (from now on just referred to as
\emph{canonical form}) at all times which is possible due to the
usage of hereditary substitution application \cite{K92,FJ25}.
As a consequence, our development does
not depend on results established for the simply-typed
$\lambda$-calculus.

Let $\xS$ be a set of \emph{sorts} ($a$, $b$).\footnote{%
This means that in the remainder,
$a$ and $b$ are used exclusively to represent arbitrary sorts.
We will continue to define naming conventions in this fashion.}
The set $\xT$ of \emph{types} ($\sigma$, $\tau$)
is defined as follows: $\xS \subseteq \xT$ and if
$(\vec{\sigma}) \in \xT^n$ and $a \in \xS$ then
$(\vec{\sigma}) \to a \in \xT$ where we identify $() \to a$ with
$a$. Furthermore, we abbreviate $(a) \to b$ to $a \to b$.
The \emph{order} of a type is defined as
$\ord{a} = 1$ and $\ord{(\vec{\sigma}) \to a} =
\m{max}\,\SET{\ord{\sigma_i} \mid 1 \leq i \leq n} + 1$
where we defined $\m{max}\,\varnothing = 0$.
Suppose there is an infinite set $\xV$ of typed variables ($x$,
$y$, $z$, $w$, $F$, $G$, $H$, $X$, $Y$, $Z$) and a set $\xF$ of typed
function symbols ($f$, $g$) called \emph{signature}, such that
$\xV$, $\xF$, $\xS$ and $\SET{\to}$ are pairwise disjoint
and there are infinitely many variables of each type. A \emph{head}
$h$ is either a function symbol or a variable and we write
$h : \sigma$ to denote that it has type $\sigma$. In order to keep
the notation concise, we drop the symbol $\lambda$ in abstractions.
The following
inference rule defines the set $\trm{\sigma}$ of \emph{terms} ($s$,
$t$, $u$, $v$) of type $\sigma$
\[
\frac{h : (\vec{\sigma}) \to a \in \xF \cup \xV \quad
(\vec{t}) \in \trm{\vec{\sigma}} \quad
(\vec[k]{x} : \vec[k]{\tau}) \in \xV^k}
{\vec[k]{x}.h(\vec{t}) \in \trm{(\vec[k]{\tau}) \to a}}
\]
where we use $\trm{\vec{\sigma}}$ as a shorthand for the set
$\trm{\sigma_1} \times \cdots \times \trm{\sigma_n}$. We abbreviate
$h()$ to $h$, $.s$ to $s$ and follow the convention that postfix
operations $\diamondsuit$ bind stronger than binders in terms, so
$\vec{x}.t\diamondsuit = \vec{x}.(t\diamondsuit)$. The set $\trm{\xS}$
denotes the set of all terms which coincides with the set of well-typed
$\lambda$-terms over $\xF$ and $\xV$ using sorts from $\xS$ in
$\beta\eta$-long normal form. As for
function symbols and variables, we denote that a term $s$ has type
$\sigma$ by $s : \sigma$. We say
that $s = \vec[k]{x}.h(\vec{s})$ is a \emph{basic} term if $k = 0$.

The set of free variables $\fv{s}$ of a term $s$ is defined
as usual.
For a sequence of terms $\vec{s}$, we use $\fv{\vec{s}}$ as
a shorthand for $\bigcup \SET{\fv{s_i} \mid 1 \leq i \leq n}$.
A term is \emph{linear} if no free variable occurs more than once in it.
Terms are uniquely defined up the renaming of bound variables
($\alpha$-renaming). Hence, we may assume that no variable occurs both
free and bound in any term.
In particular, we will reserve uppercase variable names
for free variables while lowercase variable names may represent both
bound and free variables.

Finally, note that $\xV \subseteq \trm{\xS}$ does
not hold as higher-order variables are not in canonical form. The function
$\cf$ takes a variable and returns the corresponding term
$x\cf = \vec{y}.x(\vec{y}\cf)$ where
$x : (\vec{\sigma}) \to a$ and $\vec{y} : \vec{\sigma}$.
The function $\cf$ is well-defined
by induction on the order of the involved types.
In terms of the simply-typed $\lambda$-calculus,
$\cf$ corresponds to restricted $\eta$-expansion.

\subsection{Positions and Subterms}

We refer to strings over $\mathbb{N}$ as
\emph{positions} ($p$, $q$, $r$). The
empty string, denoted by $\epsilon$, represents the
\emph{root position}. We write $pq$ for the concatenation of two
positions $p$ and $q$ and $p \leq q$ if $p$ is a prefix of $q$, i.e.,
there exists a position $r$ such that $pr = q$. In that case, we say
that $p$ is \emph{above} $q$ and $q$ is \emph{below} $p$. If neither
$p \leq q$ nor $q \leq p$ then $p$ and $q$ are \emph{parallel},
written $p \parallel q$. The set $\pos{s}$ of positions of a term
$s$ is defined inductively as follows:
$\pos{\vec[k]{x}.h(\vec{s})} = \SET{\epsilon} \cup
\SET{ip \mid \text{$1 \leq i \leq n$ and $p \in \pos{s_i}$}}$.
Given a position $p \in \pos{s}$, $s|_p$ denotes the
\emph{subterm} at position $p$ where
$s|_\epsilon = s$ and $(\vec[k]{x}.h(\vec{s}))|_{ip} = \vec[k]{x}.s_i|_p$.
In particular, the scope of bound
variables is not discarded, so the definition of subterms
is well-defined in the presence of $\alpha$-renaming. If
$s|_p = \vec[k]{x}.h(\vec{s})$ we refer to
$\vec[k]{x}$ as $\bvp{s}{p}$. Furthermore,
we use $s \subteq t$ to denote that $t$ is a subterm of
$s$ and define $s \subt t$ as $s \subteq t$ and $s \neq t$. A
position $p \in \pos{s}$ is called \emph{functional} (denoted by
$p \in \posf{s}$) if the head of $s|_p$ is a function symbol.
Finally, given $p \in \pos{s}$ and a
term $\vec{x}.t$ of the same type as $s|_p = \vec{x}.u$ we define
$s[\vec{x}.t]_p$ as the result of replacing $u$ by $t$ in $s$.
Note that this would not be well-defined if we had chosen
$\beta\eta$-normal
forms as our canonical forms due to the possible removal of
variable occurrences through subterm replacement; for
$s = x.\m{f}(\m{h}(x),x)$, $s[x.\m{c}]_1 = x.\m{f}(\m{c},x)$ is
$\eta$-reducible.

\subsection{Substitutions}

Finite mappings from variables to terms of the same type are called
\emph{substitutions} ($\theta, \gamma$, $\delta$, $\mu$). Given a
substitution $\theta = \SEt{\vec{x} \mapsto \vec{t}}$, its
\emph{domain} and \emph{image} are defined as
$\dom{\theta} = \SEt{\vec{x}}$ and $\im{\theta} = \SEt{\vec{t}}$,
respectively. The \emph{free variables} introduced by a substitution
$\theta$ are defined as the set
$\fv{\theta} = \bigcup\,\SET{\fv{t} \mid t \in \im{\theta}}$.
Given a substitution $\theta$ and a set of variables $V$, we denote by
$\theta|_V$ the restriction of $\theta$ to variables in $V$.
Furthermore, we define $\ord{\theta} =
\max\,\SET{\m{ord}(\sigma) \mid x : \sigma \in \dom{\theta}}$.
The application of a substitution
$\theta$ to a term $t$ (written in postfix notation $t\theta$) is
defined hereditarily by implicitly performing $\beta$-reduction
\cite{K92,WCPW04}:
\begin{itemize}
\item
$X(\vec{t})\theta = u\SET{\vec{x} \mapsto \vec{t}\theta}$ if
$X \mapsto \vec{x}.u \in \theta$,
\item
$h(\vec{t})\theta = h(\vec{t}\theta)$ if $h \notin \dom{\theta}$, and
\item
$(\vec{x}.t)\theta = \vec{y}.t\SEt{\vec{x} \mapsto \vec{y}}\theta$ where
$\SEt{\vec{y}}$ is a set of fresh variables.
\end{itemize}
Hence, $t\theta$ represents the
capture-avoiding application of $\theta$ to the free variables of $t$.
Note that substitution application $s\theta$ is well-defined by induction
on $(\ord{\theta},s)$ with respect to $(>_{\mathbb{N}},\subt)_{\lex}$.
Given $\theta = \SEt{\vec{x} \mapsto \vec{s}}$ and
$\delta = \SEt{\vec[m]{y} \mapsto \vec[m]{t}}$, we define their
composition as
$\theta\delta = \SET{\vec{x} \mapsto \vec{s}\delta} \cup
\SET{y_j \mapsto t_j \mid \text{$1 \leq j \leq m$ and
$y_j \notin \SEt{\vec{x}}$}}$.

A term $s$ \emph{matches} a term $t$ ($s \subsumes t$) if there exists
a substitution $\theta$ such that $s\theta = t$. In that case,
$t$ is called an \emph{instance} of $s$. If $s \subsumes t$
and $t \subsumes s$ we write $s \doteq t$ and say that $s$ and $t$
are \emph{variants}. Moreover,
if $s\theta = t\theta$ then $\theta$ is a \emph{unifier} of
$s$ and $t$ which are said to be \emph{unifiable}.
Let $W$ be a set of variables. Two substitutions
$\theta_1$ and $\theta_2$ are \emph{equal over W}
($\theta_1 = \theta_2 \restr{W}$) if $\theta_1(x) = \theta_2(x)$ for all
$x \in W$. A substitution $\theta$ is
\emph{at least as general} as a substitution $\delta$ with respect to $W$,
denoted by $\theta \subsumes \delta\restr{W}$, if there exists a
substitution $\gamma$ such that $\theta\gamma = \delta\restr{W}$.
We drop the suffix $[W]$ if $W$ is the set of all variables.
Two substitutions $\theta_1$ and $\theta_2$ are \emph{orthogonal}
($\theta_1 \perp \theta_2$) if there is no substitution $\gamma$ such
that $\theta_1 \subsumes \gamma$ and $\theta_2 \subsumes \gamma$.
Finally, given a binary relation $R$, $\theta_1 \mathrel{R} \theta_2$
denotes $\theta_1(x) \mathrel{R} \theta_2(x)$ for all
$x \in \xV$.

\subsection{Rewriting}

An ordered pair of basic terms $(\ell,r)$
can be seen as an ordered
\emph{equation} (written $\ell \approx r$), or, if we additionally
assume that $\fv{r} \subseteq \fv{\ell}$ and the head of $\ell$ is not
a variable, as a \emph{rule} (written $\ell \R r$). A
\emph{higher-order equational system} (HES) is a set of
equations while a \emph{higher-order rewrite system} (HRS) is a set of
rules. Every rule can be seen as an equation, so every HRS is also an
HES. We will exploit this by giving some general definitions
only for HESs.

Given an \textup{HES} $\xE$, its
\emph{rewrite relation} $\RbE$ is defined as follows: There is a
\emph{rewrite step} $s \RbE t$ if there exist an equation
$\ell \approx r \in \xE$, a substitution $\theta$ and a position
$p \in \pos{s}$ such that $s|_p = \vec{x}.\ell\theta$ and
$t = s[\vec{x}.r\theta]_p$. In that case, we say that $s$
is a \emph{redex} and $t$ is a \emph{reduct} of $s$.
Sometimes, we make the position $p$ (as
well as the rule $\ell \R r$ and substitution $\theta$) explicit by
writing $s \Rab{\xE}{p} t$
\smash{($s \Rab{\xE}{p \mid \ell \R r \mid \theta} t$)}. Two
equations $\ell_1 \approx r_1$ and $\ell_2 \approx r_2$ are
\emph{variants} if there exist substitutions $\theta_1$, $\theta_2$
such that $\ell_1\theta_1 = \ell_2$, $r_1\theta_1 = r_2$,
$\ell_2\theta_2 = \ell_1$
and $r_2\theta_2 = r_1$. In particular, this means that
$\ell_1 \doteq \ell_2$ and $r_1 \doteq r_2$ is witnessed by the same
substitutions.

An HRS is \emph{left-linear}
if the left-hand side of each of its rules is linear. Finally,
we say that an HRS $\xR$ is \emph{terminating} if $\Rb{\xR}$ is
well-founded and call $\xR$ (\emph{locally}) \emph{confluent} if
${\Lab{\xR}{*} \cdot \Rab{\xR}{*}} \subseteq {\Db{\xR}}$
(${\Lb{\xR} \cdot \Rb{\xR}} \subseteq {\Db{\xR}}$).

\subsection{Lifters}
\label{sec:lifters}

In higher-order rewriting, special substitutions are needed in order
to make confluence analysis through unification feasible \cite{MN98}.

\begin{definition}
A \emph{variable renaming} is a bijective function from $\xV$ to $\xV$.
Given a term $s$, a sequence of variables $\vec{x}$ and a set of variables
$W$ we define the \emph{$\vec{x}$-lifter of $s$ away from $W$} to be the
substitution $\SET{Z \mapsto \vec[m]{y}.\rho(Z)(\vec{x}\cf,\vec[m]{y}\cf) \mid
Z : (\vec[m]{\tau}) \to a \in \fv{s}}$ where $\rho$ is a variable renaming
such that $\im{\rho} \cap (W \cup \SEt{\vec{x}})
= \varnothing$, $\rho(Z) : (\vec{\sigma},\vec[m]{\tau}) \to a$,
$\vec[m]{y} : \vec[m]{\tau}$ and $\vec{x} : \vec{\sigma}$.  
\end{definition}

\begin{example}
Consider the sort $\m{a}$, the function symbols $\m{d} : \m{a}$,
$\m{f} : \m{a} \to \m{a}$ and $\m{g} : (\m{a} \to \m{a}) \to \m{a}$
as well as the variable $X : \m{a}$. For the HRS $\xR$ consisting of the
rules
\begin{align*}
\m{f}(X) &\R X & \m{g}(y_1.\m{f}(\m{f}(y_1))) &\R \m{d}
\end{align*}
we have $\m{g}(y_1.\m{f}(X)\SET{X \mapsto \m{f}(y_1)}) =
\m{g}(y_1.\m{f}(\m{f}(y_1)))$. Hence, the two rules overlap and produce
the following non-joinable local peak:
\[
\m{g}(y_1.\m{f}(y_1)) \Lab{\xR}{1} \m{g}(y_1.\m{f}(\m{f}(y_1)))
\Rab{\xR}{\epsilon} \m{d}
\]
As in first-order rewriting, we would like to use higher-order unification
on subterms of left-hand sides to compute such overlaps. However,
$y_1.\m{f}(X)$ and $y_1.\m{f}(\m{f}(y_1))$ are not unifiable
because $X$ does not depend on $y_1$. In order to resolve this,
we use the $\vec[1]{y}$-lifter $\delta = \SET{X \mapsto G(y_1)}$ where
$G : \m{a} \to \m{a}$ is a fresh variable to obtain
$y_1.\m{f}(X)\delta = y_1.\m{f}(G(y_1))$. This term is now unifiable with
$y_1.\m{f}(\m{f}(y_1))$ as witnessed by the substitution
$\SET{G \mapsto z.\m{f}(z)}$.
\end{example}

In the remainder, we often leave out the types of free
variables in examples when they can be inferred.

\section{Deterministic Higher-Order Pattern Rewrite Systems}
\label{sec:dhp}

Even though higher-order matching is decidable \cite{S09}, rewriting
with HRSs can be highly impractical.

\begin{example}
\label{exa:hrs}
Consider the sort $\m{a}$ and the function symbols
$\m{f}, \m{g} : \m{a} \to \m{a}$. The HRS
$\xR$ consisting of the single rule
$\m{f}(H(X)) \R \m{g}(X)$
is infinitely branching: For an arbitrary term $t : \m{a}$
we have $\m{f}(Y) \RbR \m{g}(t)$ by considering the
substitution $\theta_t = \SET{H \mapsto z.Y, X \mapsto t}$.
$\xR$ is not confluent as we can rewrite $\m{f}(Y)$
to two different normal forms, e.g.\ $\m{g}(X)$ and
$\m{g}(\m{g}(X))$. Moreover, the rewrite step
$\m{f}(H(X)) \R \m{g}(\m{f}(H(X)))$ using the substitution
$\SET{X \mapsto \m{f}(H(X)), H \mapsto z.H(X)}$
shows that $\xR$ is also not terminating.
\end{example}

For HRSs with nested free variables in left-hand sides like the one in
\exaref{hrs}, the variable condition
does not extend to rewriting, i.e., $\fv{r} \subseteq \fv{\ell}$ does not
imply $\fv{t} \subseteq \fv{s}$ whenever $s \Rb{\SET{\ell \R r}} t$.
Thus, in \cite{MN98}, Mayr and Nipkow only use HRSs for general
definitions and results, and concentrate their development on a
well-behaved subclass of HRSs.

\begin{definition}
A term $s$ is a \emph{pattern} if for all subterms of the form
$\vec{y}.X(\vec[m]{t})$ we have
$\SEt{\vec[m]{t}} \subseteq \SET{\vec{y}\cf}$ and
$|\SEt{\vec[m]{t}}| =  m$.
A \emph{pattern rewrite system} \textup{(PRS)} is a \textup{HRS} where
the left-hand sides of all rules are patterns. 
\end{definition}

\begin{example}
Consider the sort $\m{a}$, the function symbol
$\m{f} : \m{a} \to \m{a}$ as well as the variables
$F : (\m{a},\m{a},\m{a} \to \m{a}) \to \m{a}$ and $G : \m{a} \to \m{a}$.
The terms
\[
x,y,z.F(x,x,w.z(w)) \quad
x,y,z.F(x,\m{f}(x),w.z(w)) \quad
x,y,z.F(G(x),y,w.z(w))
\]
are not patterns but $x,y,z.F(y,x,w.z(w))$ is.
\end{example}

Patterns originate from Miller \cite{M91}.
The main motivation behind patterns is the fact that unification
of patterns is \emph{unitary}, i.e.,
unifiable patterns have most general unifiers.
This facilitates a definition of critical pairs for PRSs which closely
follows the original first-order definition \cite{MN98}.
The following example shows that PRS can easily model expressive systems
of computation.

\begin{example}
\label{exa:prs}
Consider the sort $\m{t}$ as well as the function symbols
$\m{abs} : (\m{t} \to \m{t}) \to \m{t}$
and $app : (\m{t}, \m{t}) \to \m{t}$. The PRS $\xR$ consisting of the two
rules
\begin{align*}
\m{app}(\m{abs}(z.F(z)),X) &\R F(X) & \m{abs}(z.\m{app}(X,z)) &\R X
\end{align*}
models $\beta$- and $\eta$-reduction in the untyped $\lambda$-calculus.
Given an instance of the $\eta$-reduction rule induced by a substitution
$\theta$, the side condition $z \notin \fv{\theta(X)}$ is implicit 
by the usage of capture-avoiding substitution.
\end{example}

In order to eliminate the problems inherent in \exaref{hrs},
one does not have to demand unitary unification of left-hand sides.
In fact, it is enough to contain matching to a finite problem.
Analogous to unitary unification, it is very convenient have
\emph{deterministic} matching, i.e., given a concrete term and position,
a rule is either applicable (with a unique substitution) or not. Since
pattern unification is unitary, matching with patterns is also
deterministic.
In this section, we explore a way of generalizing patterns such
that deterministic matching is preserved. As for patterns, this class of
terms will then serve as a basis for a new class of HRSs by restricting
the left-hand sides of rules accordingly.
The consequences for unification as well as critical pairs shall only
trouble us in the next section.

As we have already seen in \exaref{hrs}, nested free variables will not
be allowed in the left-hand sides of our rules. The following example
shows that this restriction is not enough.

\begin{example}
\label{exa:const}
Consider the sort $\m{a}$ as well the function symbols
$\m{f} : \m{a} \to \m{a}$ and $\m{c},\m{d} : \m{a}$.
With the HRS $\xR$ consisting of the single rule
$\m{f}(H(\m{c})) \R H(\m{d})$
there are two ways to rewrite the term $\m{f}(\m{c})$ at the root position.
Using $\SET{H \mapsto z.z}$ we obtain $\m{f}(\m{c}) \RbR \m{d}$
while $\SET{H \mapsto z.\m{c}}$ induces the rewrite step
$\m{f}(\m{c}) \RbR \m{c}$.
\end{example}

The underlying problem of \exaref{const} is that the argument $\m{c}$ of
$H$ can also be used in the substitution. If the argument contains a bound
variable, this is no longer possible due to capture-avoidance in
substitution application:

\begin{example}
\label{exa:captavoid}
Consider the sort $\m{a}$ as well as the function symbols
$\m{f} : (\m{a} \to \m{a}) \to \m{a}$,
$\m{g} : \m{a} \to \m{a}$ and $\m{c} : \m{a}$.
With the HRS $\xR$ consisting of the single rule
\[
\m{f}(z.H(\m{g}(z))) \R H(\m{c})
\]
the term $\m{f}(z.\m{g}(z))$ can only be rewritten to $\m{c}$.
Here, the employed substitution is $\SET{H \mapsto y.y}$.
\end{example}

To sum up, the class of terms we are looking for already has the following
constraints: Arguments of free variables
may not contain free variables and
have to contain at least one bound variable in order to force a matching
substitution to use it.
Until now, we have only looked at higher-order
variables which take one argument. The following example sheds light on
the conditions we have to impose in the general case.

\begin{example}
\label{exa:localcond}
Consider the sort $\m{a}$ as well as the function symbols
$\m{f} : (\m{a} \to \m{a}) \to \m{a}$,
$\m{g} : \m{a} \to \m{a}$ and $\m{c} : \m{a}$.
With the HRS $\xR$ consisting of the single rule
\[
\m{f}(z.H(y.z(y),z(\m{c}))) \R H(y.\m{g}(y),\m{c})
\]
there are two ways to rewrite the term $\m{f}(z.z(\m{c}))$ at the
root position: Using $\SET{H \mapsto y_1,y_2.y_1(\m{c})}$ we obtain
$\m{f}(z.z(\m{c})) \R \m{g}(\m{c})$ while
$\SEt{H \mapsto y_1,y_2.y_2}$ induces the rewrite step
$\m{f}(z.z(\m{c})) \R \m{c}$.
\end{example}

Intuitively, we have to demand that arguments of a free variable cannot
be constructed from the remaining arguments and function symbols.
This can be seen as a generalization of the
distinctness condition for Miller's patterns to our setting. We are now very
close to the definition of deterministic higher-order patterns (DHPs)
due to Yokoyama et al.~\cite{YHT04}. The only additional constraint
on DHPs is that the $\eta$-normal forms of the arguments of
free variables are not abstractions. This greatly simplifies the
complexity of dealing with DHPs as it allows us to replace the
vague description of ``not constructible'' by a straightforward
overapproximation which demands that arguments of free variables are not
higher-order subterms of each other. Furthermore, this constraint
also facilitates the matching algorithm
to run in linear time \cite{YHT04}.

\subsection{Definition}

We will now adapt the definition of DHPs to our definition of terms.
To that end, we introduce the notion of expanded terms
to represent the terms whose $\eta$-normal form is not an abstraction.
For these terms, the expanded subterm relation corresponds to
the higher-order subterm relation.

\begin{definition}
A term $\vec{x}.s$ is \emph{expanded} if
$\vec{x}.s = \vec{x},\vec[k]{y}.h(\vec[m]{s},\vec[k]{y}\cf)$ such that
$\bigl(\SET{h} \cup \fv{\vec[m]{s}}\bigr) \cap \SEt{\vec[k]{y}} =
\varnothing$. Consider a term $\vec{x}.s$ as well as an expanded term
$\vec{x}.t = \vec{x},\vec[k]{y}.h(\vec[m]{t},\vec[k]{y}\cf)$. We say
that $\vec{x}.t$ is an \emph{expanded subterm} of $\vec{x}.s$
$(\vec{x}.s \subtEeq \vec{x}.t)$ if there exist
$n' \geq n$ and $k$ terms written as
$\vec[n']{x}.t_{m+1},\dots,\vec[n']{x}.t_{m+k}$
such that $\vec{x}.s \subteq \vec[n']{x}.h(\vec[m+k]{t})$.
We write $\vec{x}.s \subtE \vec{x}.t$ if $\vec{x}.s \subtEeq \vec{x}.t$ and
$\vec{x}.s \neq \vec{x}.t$.
\end{definition}

\begin{example}
Consider the sort $\m{a}$ as well as the function symbols $\m{c} : \m{a}$,
$\m{g} : (\m{a},\m{a}) \to \m{a}$ and
$\m{f} : (\m{a},\m{a},\m{a} \to \m{a}) \to \m{a}$.
The terms $s_1 = \m{c}$, $s_2 = \m{g}(X,Y)$, $s_3 = x.\m{g}(\m{c},x)$ and
$x.s_4 = x,y.\m{f}(x,\m{c},z.y(z))$ are expanded
but $s_5 = x.\m{g}(\m{c},Y)$ and $x.s_6 = x,y.\m{f}(z.y(z),\m{c},z.y(z))$
are not. Consider $x.u = x.\m{f}(x,\m{c},y.\m{g}(x,y))$. We have
$x.u \subtEeq x.v_i$ for $1 \leq i \leq 3$ where
$x.v_1 = x,z_1,z_2,z_3.\m{f}(z_1,z_2,z_4.z_3(z_4))$,
$x.v_2 = x,z.\m{f}(x,\m{c},y.z(y))$ and
$x.v_3 = x,z.\m{g}(x,z)$.
On the other hand, $x.u \nsubtEeq x.v_4$ where
$x.v_4 = x,y.\m{f}(\m{c},x,z.y(z))$.
\end{example}

For expanded terms $\vec{x}.t$, $s \subteq \vec{x}.t$ implies
$s \subtEeq \vec{x}.t$.
Also note that in the definition of $\subtEeq$ we do
not have to demand that $x_{n+1},\dots,x_{n'}$ are fresh
as these variables occur as bound variables in $\vec{x}.s$
and we follow the convention that no variable occurs both
free and bound in a term.
We are now ready to give a formal definition of DHPs.

\begin{definition}[adapted from \cite{YHT04}]
\label{def:dhp}
A term $s$ is a
\emph{deterministic higher-order pattern} \textup{(DHP)} if the
following conditions hold for all subterms
$\vec[n]{y}.X(\vec[m]{t})$ and $1 \leq i \leq m$:
\begin{enumerate}[(i)]
\item \label{enu:dhp-var}
$\varnothing \neq \fv{t_i} \subseteq \SEt{\vec{y}}$,
\item \label{enu:dhp-exp}
$\vec{y}.t_i$ is an expanded term,
\item \label{enu:dhp-subt}
$\vec{y}.t_i \nsubtEeq \vec{y}.t_j$ whenever $i \neq j$.
\end{enumerate}
We refer to the terms $\vec[m]{t}$ as \emph{anchors} of
the \textup{DHP}.
\end{definition}

\begin{example}
Consider the sort $\m{a}$, the function symbols $\m{f} : \m{a} \to \m{a}$,
$\m{c} : \m{a}$ and
the variables $F : (\m{a},\m{a}) \to \m{a}$ and
$G : (\m{a},\m{a} \to \m{a}) \to \m{a}$.
The terms $x,y.F(y,x)$, $x,y.F(\m{f}(x),\m{f}(y))$ and
$x,y.F(x(y),x(\m{c}))$ are DHPs.
Only the first term is also
a pattern. On the other hand,
none of the terms $x,y.F(\m{c},x)$, $x,y.G(y,z.x(z,\m{c}))$,
$x,y.F(x,x)$, $x,y.F(\m{f}(x),x)$, and $x,y.G(x(y),z.x(z))$ is a
DHP because the first one violates condition \enuref{dhp-var}, the
second one violates condition \enuref{dhp-exp}
and the remaining terms violate condition \enuref{dhp-subt}.
\end{example}

\begin{theorem}[Theorem 3 from \cite{YHT04}]
\label{thm:matchUnique}
Let $s$ be a \textup{DHP} and $t$ an arbitrary term.
If $s$ matches $t$ then this is witnessed by exactly one substitution.
\end{theorem}

We are now ready to define a subclass of HRSs based on DHPs. A
\emph{deterministic higher-order pattern rewrite rule} is a rewrite rule
whose left-hand side is a DHP. A \emph{deterministic higher-order pattern
rewrite system} \textup{(DPRS)} is a set of deterministic higher-order
pattern rewrite rules. Since DHPs strictly include patterns, the
set of DPRSs strictly includes the set of PRSs.
To the best of our knowledge, DPRSs have not been
considered in the literature.

\subsection{Examples}

We conclude this section by showcasing three typical use cases of
DPRSs which cannot be directly achieved with PRSs. The most basic
one is the usage of DHPs to force the argument of a function
to be in one of a collection of fixed contexts.

\begin{example}
\label{exa:round}
Consider the sorts $\m{int}$, $\m{double}$ as well as the function
symbols $\m{id} : \m{int} \to \m{int}$,
$\m{fromInt} : \m{int} \to \m{double}$,
$\m{round} : \m{double} \to \m{int}$ and
$\m{app} : (\m{int} \to \m{double}, \m{double}) \to \m{double}$.
The DPRS $\xR$ consisting of the following rewrite rules
\begin{align*}
\m{app}(z.F(\m{id}(z),\m{fromInt}(z)),X) &\R F(\m{round}(X),X) \\
\m{round}(\m{fromInt}(Y)) &\R \m{id}(Y)
\end{align*}
makes use of DHPs in order to compute the application
of a function $g$ from $\m{int}$ to $\m{double}$ to a term $t$ of type
$\m{double}$. If the argument of $g$ is used as an $\m{int}$,
we have to cast $t$ to $\m{int}$ by applying $\m{round}$.
However, if the argument of $g$ is used as a $\m{double}$ (indicated
by wrapping it with $\m{fromInt}$), we can use $t$ directly.
Note that the usage of $z.\m{id}(z)$ instead of just $z.z$
in the left-hand side of the first rule is necessary since it
would not be a DHP otherwise.
\end{example}

The second typical use case of DPRSs is to check whether a given
higher-order expression can be matched by a DHP in order to
construct a function using higher-order matching.
We now present a DPRS modeling the
the motivating example in the original paper on DHPs \cite{YHT04}.
In particular, all computations in the original example can be
expressed by a single rewrite sequence.
In \secref{progtrans}, we will argue that the upcoming example
is just a glimpse into the applicability of DPRSs for
a large class of program transformations in functional programming.

\begin{example}
\label{exa:fusion}
Consider the sorts $\m{a}$, $\m{b}$ as well as the function
symbols given in \tabref{fusionfs}.
The DPRS $\xR$ consisting of the following rewrite rules
\begin{align*}
&\m{sum}(\m{nil}) \R \m{0} \\
&\m{sum}(\m{cons}(X,Y)) \R X + \m{sum}(Y) \\
&\m{comp}(z_3.\m{sum}(z_3),z_0.\m{foldr}_1(z_1,z_2.G(z_1,z_2),Y,z_0),Z)
\\
&\quad \R \m{check}(z_1,z_2.\m{sum}(G(z_1,z_2)),Y,Z) \\
&\m{check}(z_1,z_2.F(z_1,\m{sum}(z_2)),X,Y) \\
&\quad \R \m{foldr}_2(z_1,z_2.F(z_1,z_2),\m{sum}(X),Y)
\end{align*}
makes use of DHPs to encode a basic fusion transformation rule in
functional programming. In particular, the rewrite sequence
\begin{align*}
z_0.&\m{comp}(z_1.\m{sum}(z_1),
z_1.\m{foldr}_1(z_2,z_3.\m{cons}(z_2 \times z_2,z_3),\m{nil},z_1),z_0) \\
&\RbR z_0.\m{check}(z_1,z_2.\m{sum}(\m{cons}(z_1 \times z_1,z_2)),
\m{nil},z_0) \\
&\RbR z_0.\m{check}(z_1,z_2.z_1 \times z_1 + \m{sum}(z_2)),\m{nil},z_0) \\
&\RbR z_0.\m{foldr}_2(z_1,z_2.z_1 \times z_1 + z_2,\m{sum}(\m{nil}),z_0) \\
&\RbR z_0.\m{foldr}_2(z_1,z_2.z_1 \times z_1 + z_2,\m{0},z_0)
\end{align*}
fully captures the transformation described in the motivating example of
\cite{YHT04}: The composition of the $\m{sum}$ function with a
$\m{foldr}$ expression mapping a list of numbers to its squares can be
transformed into a $\m{foldr}$ expression which directly computes the sum
of squares of a given list of numbers. Here, $\m{check}$ is used
to model the side condition of the actual rewrite step from the left-hand
side of the third rule to the right-hand side of the fourth rule:
If the first argument of $\m{check}$ can be rewritten (by applying
the definition of $\m{sum}$) such that it is an instance of
the corresponding DHP in the left-hand side of the fourth rule,
the computation succeeds by using the function constructed in
the DHP matching process for the final $\m{fold}$ expression.
The transformation from conditional rewriting
to ordinary rewriting through the auxiliary function $\m{check}$
is called \emph{unraveling} \cite{M96}.
\end{example}

\begin{table}[t]
\centering
\caption{Signature for the DPRS $\xR$ in \exaref{fusion}.}
\label{tab:fusionfs}
\begin{tabular}{c@{\quad}c@{\qquad}c@{\quad}c}
symbol & type & symbol & type \\ \midrule
$\m{0}$ & $\m{a}$ &
$\m{foldr}_1$ & $((\m{a}, \m{b}) \to \m{b}, \m{b}, \m{b}) \to \m{b}$ \\
${+}, {\times}$ & $(\m{a}, \m{a}) \to \m{a}$ &
$\m{foldr}_2$ & $((\m{a}, \m{a}) \to \m{a}, \m{a}, \m{b}) \to \m{a}$ \\
$\m{nil}$ & $\m{b}$ &
$\m{comp}$ & $(\m{b} \to \m{a},\m{b} \to \m{b}, \m{b}) \to \m{a}$ \\
$\m{cons}$ & $(\m{a}, \m{b}) \to \m{b}$ &
$\m{check}$ & $((\m{a}, \m{b}) \to \m{a}, \m{b}, \m{b}) \to \m{a}$ \\
$\m{sum}$ & $\m{b} \to \m{a}$
\end{tabular}
\end{table}

The third typical use case of DPRSs is the extraction of some
common context in the left-hand side. The following example
is taken from \cite[Section 5.1.2.2]{vdP96}. A variant of
it is also included in Nipkow's original paper on
higher-order critical pairs \cite{N91} where it is already stated
that an extension of the theory to accommodate such rewrite
systems is ``highly desirable''.

\begin{example}
\label{exa:sdu}
Consider the sorts $\m{a}$, $\m{b}$, $\m{u}$ as well as the function
symbols $\m{inl} : \m{a} \to \m{u}$, $\m{inr} : \m{b} \to \m{u}$ and
$\m{case}_\iota : (\m{u},\m{a} \to \iota,\m{b} \to \iota) \to \iota$
where $\iota \in \SET{\m{a},\m{b},\m{u}}$. From the rule schemas
\begin{align*}
\m{case}_\iota(\m{inl}(X),z_1.F(z_1),z_2.G(z_2)) &\R F(X) \\
\m{case}_\iota(\m{inr}(Y),z_1.F(z_1),z_2.G(z_2)) &\R G(Y) \\
\m{case}_\iota(Z,z_1.H(\m{inl}(z_1)),z_2.H(\m{inr}(z_2))) &\R H(Z)
\end{align*}
we obtain the DPRS $\xR$ consisting of 9 rules which models
a case distinction operation on a surjective disjoint union datatype
$\m{u} = \m{a} \uplus \m{b}$. Here, $\m{inl}$ ($\m{inr}$) lifts terms of
type $\m{a}$ ($\m{b}$) into the disjoint union datatype $\m{u}$.
\end{example}

\section{Critical Pairs}
\label{sec:cps}

\subsection{Minimal Complete Sets of Unifiers}

The crucial ingredient of any definition of critical pairs is a
corresponding unification procedure. As opposed to
Miller's pattern unification \cite{M91} or FCU \cite{LM22}, unification of
DHPs is not unitary. As a consequence, critical pairs in DPRSs are not
computed from most general unifiers, but from minimal complete sets of
unifiers.

\begin{definition}
A \emph{minimal complete set of unifiers} $U$ of two \textup{DHP}s $s$
and $t$ satisfies the following conditions:
\begin{enumerate}[(i)]
\item
$\dom{\theta} \subseteq \fv{s} \cup \fv{t}$ for all $\theta \in U$,
\item
$s\theta = t\theta$ for all $\theta \in U$,
\item
if $s\delta = t\delta$ then
$\theta \subsumes \delta\restr{\fv{s} \cup \fv{t}}$
for some $\theta \in U$,
\item
if $\theta_1,\theta_2 \in U$ and $\theta_1 \neq \theta_2$ then
$\theta_1 \perp \theta_2$.
\end{enumerate}
\end{definition}

\begin{example}[taken from \cite{YHT04}]
Consider the sort $\m{a}$, the function symbol $\m{f} : \m{a} \to \m{a}$
as well as the variables $F,G : (\m{a},\m{a}) \to \m{a}$. The DHPs
$x,y.F(\m{f}(x),\m{f}(y))$ and $x,y.\m{f}(G(y,x))$ admit the following
minimal complete set of unifiers
\begin{itemize}
\item
$\SEt{F \mapsto z_1,z_2.z_1, G \mapsto z_1,z_2.z_2}$,
\item
$\SEt{F \mapsto z_1,z_2.z_2, G \mapsto z_1,z_2.z_1}$, and
\item
$\SET{F \mapsto z_1,z_2.\m{f}(Z(z_1,z_2)),
G \mapsto z_1,z_2.Z(\m{f}(z_2),\m{f}(z_1))}$
\end{itemize}
where $Z$ is a fresh variable.
\end{example}

We have presented a sound and complete inference system for computing minimal complete sets
of unifiers of DHPs in \cite{NM26a}.
Its most notable difference from the standard inference system for full
higher-order unification \cite{SG89} is that flex-flex pairs can be
solved in a most general way. Unfortunately, minimal complete sets of
DHP unifiers may still be infinite. For critical pair computation,
we only need unification of DHPs which do not share free variables. The
following example shows that minimal complete sets of unifiers can be
infinite even in this special setting.

\begin{example}
Consider the sort $\m{a}$, the function symbols
$\m{f} : \m{a} \to \m{a}$ and $\m{g} : (\m{a},\m{a} \to \m{a}) \to \m{a}$
as well as the variables $F, G : \m{a} \to \m{a}$. The DHPs
$\m{g}(F(x),x.F(\m{f}(x)))$ and $\m{g}(G(x),x.\m{f}(G(x)))$ admit the
infinite minimal complete set of unifiers $\SET{\theta_k \mid k \geq 0}$
where $\theta_k = \SET{F \mapsto x.\m{f}^k(x), G \mapsto x.\m{f}^k(x)}$.
Here, $\m{f}^0(t) = t$ and $\m{f}^{n+1} = \m{f}(\m{f}^n(t)).$
\end{example}

In the following novel result we show that minimal complete sets of
unifiers are guaranteed to be finite if we additionally demand linearity
of the involved terms. This is relevant as it ensures that finite
left-linear DPRSs have finitely many critical pairs.

\begin{restatable}{theorem}{unifdecidable}
\label{thm:unifdecidable}   
For linear \textup{DHPs} $s$ and $t$ with
$\fv{s} \cap \fv{t} = \varnothing$,
minimal complete sets of unifiers are finite and computable.
\end{restatable}

\subsection{Definition}

For convenience, we will extract critical pairs from overlaps which
carry all necessary information for computing the critical pairs.
We obtain a preliminary definition of overlaps of DPRSs
by replacing most general unifiers with minimal complete
sets of unifiers in the definition of overlaps for PRSs \cite{MN98}.
For each overlap, we will obtain one critical pair per
unifier in the minimal complete set of unifiers.

\begin{definition}[preliminary]
\label{def:overlapPreliminary}
An \emph{overlap} of a \textup{DPRS} $\xR$ is a sextuple
\[\langle \ell_1 \R r_1,p,(\vec{x}),\delta,U,\ell_2 \R r_2 \rangle\]
satisfying the following properties:
\begin{enumerate}[(i)]
\item
$\ell_1 \R r_1, \ell_2 \R r_2 \in \xR$,
\item \smallskip
$p \in \posf{\ell_2}$, $\bvp{\ell_2}{p} = \vec{x}$ and $\delta$ is an
$\vec{x}$-lifter of $\ell_1$ away from $\fv{\ell_2}$,
\item \smallskip
$U \neq \varnothing$ is a minimal complete set of unifiers of
$\vec{x}.\ell_1\delta$ and $\ell_2|_p$,
\item \smallskip
if $p = \epsilon$ then $\ell_1 \R r_1$ and $\ell_2 \R r_2$ are not
variants.
\end{enumerate}
Given an overlap
$\langle \ell_1 \R r_1,p,\vec{x},\delta,U,\ell_2 \R r_2 \rangle$
of a \textup{DPRS} $\xR$ and $\theta \in U$ we obtain the
\emph{critical peak}
\[
\ell_2\theta[(\vec{x}.r_1\delta)\theta]_p \Lab{\xR}{p}
\ell_2\theta[(\vec{x}.\ell_1\delta)\theta]_p =
\ell_2\theta \Rab{\xR}{\epsilon} r_2\theta
\]
which induces the \emph{critical pair}
$\ell_2\theta[(\vec{x}.r_1\delta)\theta]_p \approx r_2\theta$.
\end{definition}

We start by illustrating the previous definition on a well-known PRS.

\begin{example}
\label{exa:prscp}
The PRS $\xR$ from \exaref{prs} contains two overlaps. The overlap
\[
\langle \m{abs}(y.\m{app}(X,y)) \R X,1,(),\delta_1,U_1,
\m{app}(\m{abs}(z.F(z)),X) \R F(X) \rangle
\]
where $\delta_1 = \SET{X \mapsto Y}$ and
$U_1 = \SET{F \mapsto w.\m{app}(Y,w)}$
gives rise to the critical peak
\[
\m{app}(Y,X) \Lab{\xR}{1} \m{app}(\m{abs}(z.\m{app}(Y,z)),X)
\Rab{\xR}{\epsilon} \m{app}(Y,X)
\]
while the overlap
\[
\langle \m{app}(\m{abs}(y.F(y)),X) \R F(X),1,z,\delta_2,U_2,
\m{abs}(z.\m{app}(X,z)) \R X \rangle
\]
where $\delta_2 = \SET{F \mapsto w.G(z,w),X \mapsto Z(z)}$
and
\begin{align*}
U_2 &= \SET{X \mapsto \m{abs}(y.H(y)), G \mapsto y_1,y_2.H(y_2),
Z \mapsto y.y}
\end{align*}
gives rise to the following critical peak:
\[
\m{abs}(z.H(z)) \Lab{\xR}{1} \m{abs}(z.\m{app}(\m{abs}(y.H(y)),z))
\Rab{\xR}{\epsilon} \m{abs}(y.H(y))
\]
The corresponding critical pairs are trivially joinable, so $\xR$ (and
therefore the untyped $\lambda$-calculus) is locally confluent. 
\end{example}

\begin{example}
\label{exa:overlapPreliminary}
Consider the sort $\m{a}$
as well as the function symbols
$\m{c} : (\m{a} \to \m{a}) \to \m{a}$, $\m{d} : \m{a}$ and
$\m{f}, \m{g} : \m{a} \to \m{a}$. The
DPRS $\xR$ consisting of the rewrite rules
\begin{align*}
1\colon \quad \m{f}(\m{g}(X)) &\R \m{f}(X) &
2\colon \quad \m{c}(y.Z(\m{g}(y))) &\R Z(\m{d})
\end{align*}
does not have critical pairs according to \defref{overlapPreliminary}.
However, the local peak
\[
\m{c}(y.\m{f}(y)) \Lab{\xR}{1} \m{c}(y.\m{f}(\m{g}(y)))
\Rab{\xR}{\epsilon} \m{f}(\m{d})
\]
is not joinable. Hence, joinability of critical pairs as defined in
\defref{overlapPreliminary} does not imply local confluence of DPRSs.
\end{example}

As already observed by Hamana \cite{H19}, one also has to consider
overlaps at variable positions when going beyond patterns.
To the best of our knowledge,
Hamana provides the sole existing definition of critical
pairs in such a setting, albeit for the restricted class of
second-order systems. We will now show that
Hamana's confluence result \cite[Theorem 7.11]{H19} is actually unsound.

\begin{example}
\label{exa:hamana1}
Recall the DPRS $\xR$ from \exaref{overlapPreliminary}.
According to Hamana's definition, the non-joinable local
peak exhibited in the example is a critical peak which
gives rise to the only critical pair
$\m{c}(y.\m{f}(y)) \approx \m{f}(\m{d})$ of $\xR$.
Following the same definition, the
DPRS $\xR' = \xR \cup \SET{\m{c}(y.\m{f}(y)) \R \m{f}(\m{d})}$
still only has one critical pair which is now joinable. Hence, according
to Hamana, $\xR'$ is locally confluent but actually this is not the
case as now the following local peak is not joinable:
\[
\m{c}(y.\m{f}(\m{f}(y))) \Lab{\xR'}{11} \m{c}(y.\m{f}(\m{f}(\m{g}(y))))
\Rab{\xR'}{\epsilon} \m{f}(\m{f}(\m{d}))
\]
\end{example}

The main challenge of establishing a critical
pair lemma for DPRSs is to give a complete characterization of
overlaps at variable positions which yields at most one critical pair
per unifier in the minimal complete set of unifiers. In the
upcoming example, we will gradually refine the notion of variable overlap
until we are ready to give a correct definition.

\begin{example}
\label{exa:varOverlap}
Recall once more the DPRS $\xR$ from \exaref{overlapPreliminary}.
As expected, we obtain
$\m{c}(y.\m{f}(\m{g}(y)))$
by unifying $y.Z(\m{g}(y))$ and
the lifted term $y.\m{f}(\m{g}(X(y)))$
with the substitution
$\theta = \SET{Z \mapsto z.\m{f}(z), X \mapsto z.z}$.\footnote{%
Note that due to the global restriction, FCU \cite{LM22} is not
applicable here.}
In order to resolve the issue in \exaref{hamana1},
it is natural to wrap the result of the unification
into a fresh variable $X : \m{a} \to \m{a}$ which can be instantiated with
the function corresponding to the concrete context of the actual rewrite
step in any given instance of $y.Z(\m{g}(y))$. Hence, we arrive at
\[
\m{c}(y.X(\m{f}(y))) \Lb{\xR} \m{c}(y.X(\m{f}(\m{g}(y)))) \Rb{\xR}
X(\m{f}(\m{d}))
\]
as our new candidate for the critical peak corresponding to the
overlap. Still, there is the non-joinable local peak
\[
\m{c}(y.F(\m{f}(y),\m{g}(y))) \Lab{\xR}{11}
\m{c}(y.F(\m{f}(\m{g}(y)),\m{g}(y))) \Rab{\xR}{\epsilon}
F(\m{f}(\m{d}),\m{d})
\]
caused by the same overlap. Here, the problem is that
$X$ cannot reconstruct the used context since
it may only use the bound variable $y$ when it is wrapped inside an
$\m{f}$. Hence, in general, we have to supply
the original arguments of $Z$ to $X$.
We arrive at the following improved candidate for the critical peak
corresponding to the overlap:
\[
\m{c}(y.Y(\m{f}(y),\m{g}(y))) \Lab{\xR}{11}
\m{c}(y.Y(\m{f}(\m{g}(y)),\m{g}(y))) \Rab{\xR}{\epsilon}
Y(\m{f}(\m{d}),\m{d})
\]
The next problem is that $\SET{\theta}$ is actually not a minimal complete
set of unifiers of $y.Z(\m{g}(y))$ and $y.\m{f}(\m{g}(X(y)))$.
There is also the unifier
\[
\SET{Z \mapsto z.\m{f}(\m{g}(H(z))), X \mapsto z.H(\m{g}(z))}
\]
but the corresponding local peak
\[
\m{c}(y.Y(\m{f}(H(\m{g}(y))),\m{g}(y))) \Lab{\xR}{11}
\m{c}(y.Y(\m{f}(\m{g}(H(\m{g}(y)))),\m{g}(y))) \Rab{\xR}{\epsilon}
Y(\m{f}(\m{g}(H(\m{d}))),\m{d})
\]
can always easily be joined since the unifier contains an instance of the
left-hand side of rule 1:
\[
\m{c}(y.Y(\m{f}(H(\m{g}(y))),\m{g}(y))) \RbR Y(\m{f}(H(\m{d})),\m{d}) \LbR
Y(\m{f}(\m{g}(H(\m{d}))),\m{d})
\]
Hence, we can safely ignore this unifier.
The local peak
\[
\m{c}(y.Y(G(\m{d},\m{g}(y)),\m{g}(y)))
\Lab{\xR}{11}
\m{c}(y.Y(\m{c}(z.G(\m{g}(z),\m{g}(y))),\m{g}(y)))
\Rab{\xR}{\epsilon}
Y(\m{c}(z.G(\m{g}(z),\m{d})),\m{d})
\]
corresponding to the variable
overlap between rule 2 and itself at position 1 can be safely ignored
for the same reason:
\[
\m{c}(y.Y(G(\m{d},\m{g}(y)),\m{g}(y))) \RbR Y(G(\m{d},\m{d}),\m{d}) \LbR
Y(\m{c}(z.G(\m{g}(z),\m{d})),\m{d})
\]
Note that superfluous variable overlaps between rules and
themselves are a general phenomenon.
Next, suppose we extend $\xR$ with the sort $\m{b}$,
function symbols $\m{e} : \m{b} \to \m{a}$ and $\m{h} : \m{a} \to \m{b}$
as well as the rule
\[
3\colon\quad \m{h}(\m{g}(X)) \R \m{h}(X)
\]
and consider the non-joinable local peak
\[
\m{c}(y.\m{e}(\m{h}(y))) \Lab{\xR}{11} \m{c}(y.\m{e}(\m{h}(\m{g}(y))))
\Rab{\xR}{\epsilon} \m{e}(\m{h}(\m{d}))
\]
caused by a variable overlap of rule 3 with rule 2 at position 1.
Unifying $y.Z(\m{g}(y))$  and the lifted term $y.\m{h}(\m{g}(X(y)))$
is not possible because the types do not match. Hence, before
attempting unification for a variable overlap, we have to replace
the variable by a fresh variable of appropriate type
in order to account for all possible overlaps
at deeper positions in the original variable.
Finally, suppose we extend $\xR$ with the function symbols
$\m{j} : (\m{a},\m{a}) \to \m{a}$ and
$\m{i}: (\m{a} \to \m{a}) \to \m{a}$
as well as the following rule:
\[
4\colon\quad \m{j}(\m{g}(X_1),X_2) \R X_2
\]
The non-joinable local peak
\[
\m{c}(y.\m{i}(z.z)) \Lab{\xR}{11} \m{c}(y.\m{i}(z.\m{j}(\m{g}(y),z)))
\Rab{\xR}{\epsilon} \m{i}(z.\m{j}(\m{d},z))
\]
caused by a variable overlap of rule 4 with rule 2 at position 1 is
not an instance of the corresponding critical peak
\[
\m{c}(y.Y(X_2,\m{g}(y))) \Lab{\xR}{}
\m{c}(y.Y(\m{j}(\m{g}(y),X_2),\m{g}(y))) \Rab{\xR}{\epsilon}
Y(\m{j}(\m{d},X_2),\m{d})
\]
because the first argument of $Y$ cannot be made functionally dependent on
the bound variable $z$ introduced in the instance of $Y$. This can be
resolved by using the following \emph{lifted} critical peak
\[
\m{c}(y.Y'(x.X_2(x),\m{g}(y))) \Lab{\xR}{}
\m{c}(y.Y'(x.(\m{j}(\m{g}(y),X_2(x)),\m{g}(y)))) \Rab{\xR}{\epsilon}
Y'(x.\m{j}(\m{d},X_2(x)),\m{d})
\]
as witnessed by the substitution
$\SET{Y' \mapsto y_1,y_2.\m{i}(z.y_1(z)), X_2 \mapsto y.y}$.
\end{example}

The upcoming definition of overlaps extends our preliminary
definition by an additional position $q$, a substitution
$\gamma$ as well as a lifting bound variable context $\vec[l]{w}$
in order to give a unified presentation of overlaps.
Intuitively, if $p \in \posf{\ell_2}$ then $q = \epsilon$,
$\gamma = \varnothing$, $l = 0$ and the definition coincides with
\defref{overlapPreliminary}. Otherwise we set
$q = 1$ and use $\gamma$ and $\vec[l]{w}$ to model the general
case of a variable overlap. Note that $\vec[l]{w}$ is not determined
by the rules but by the concrete local peak under consideration.

\begin{definition}
\label{def:overlap}
An \emph{overlap} of a \textup{DPRS} $\xR$ is a nonuple
\[
\langle \ell_1 \R r_1,p,q,(\vec{x}),(\vec[l]{w}),
\delta,\gamma,U,\ell_2 \R r_2 \rangle
\]
satisfying the following properties:
\begin{enumerate}[(i)]
\item
$\ell_1 \R r_1, \ell_2 \R r_2 \in \xR$,
\item \smallskip
$p \in \pos{\ell_2}$, $\bvp{\ell_2}{p} = \vec{x}$ and $\delta$ is an
$\vec{x},\vec[l]{w}$-lifter of $\ell_1$ away from
$\fv{\ell_2\gamma}$,
\item \smallskip
$U' \neq \varnothing$ is a minimal complete set of unifiers of
$\vec{x},\vec[l]{w}.\ell_1\delta$ and $\ell_2\gamma|_{pq}$,\footnote{%
By definition, $\ell_2\gamma|_{pq}$ is a DHP, and
$\vec{x},\vec[l]{w}.\ell_1\delta$ is a DHP because being a DHP is
preserved by the application of a lifter.}
\item \smallskip
\label{enu:overlap-iv}
either (a) $q = \epsilon$, $l = 0$, $\gamma = \varnothing$,
$p \in \posf{\ell_2}$ and $U = U'$, or (b) $q = 1$ and
\begin{itemize}
\item \smallskip
$\ell_2|_p = \vec{x}.Y(\vec[m]{s})$,
\item \smallskip
$\gamma = \SET{Y \mapsto \vec[m]{y}.Y''(\vec[l]{w}.Y'(\vec[m]{y}\cf
,\vec[l]{w}\cf),\vec[m]{y}\cf)}$
with
$Y' : (\vec[m]{\sigma},\vec[l]{\tau}) \to b$ and
$Y'' : ((\vec[l]{\tau}) \to b,\vec[m]{\sigma}) \to a$
fresh variables where
$Y : (\vec[m]{\sigma}) \to a$ and $\ell_1 : b$,
\item \smallskip
$U = U' \setminus U''$ such that $U''$ is the set of all $\theta \in U'$
for which at least one of the following holds:
\begin{enumerate}[(1)]
\item \smallskip
 $\theta(Y') = \vec[m]{z},\vec[l]{w}.z_i\cf$ for some
$1 \leq i \leq m$,
\item \smallskip
$\theta(Y') = \vec[m]{z},\vec[l]{w}.\ell_1\mu$ for some $\mu$,
\end{enumerate}
\end{itemize}
\item \smallskip
if $p = \epsilon$ then $\ell_1 \R r_1$ and $\ell_2 \R r_2$ are not
variants.
\end{enumerate}
\end{definition}

Condition \enuref{overlap-iv}(b)(1)
makes sure that overlaps at anchors
are not considered again as variable overlaps.
Furthermore, as explained in \exaref{varOverlap},
condition \enuref{overlap-iv}(b)(2)
rules out trivial variable overlaps where $\theta(Y')$ is an
``instance'' of $\ell_1$.
Note that variable overlaps are always trivial
if the occurrence of $Y$ at position $p$ in $\ell_2$ only has bound
variables as arguments. In particular, this holds for Miller's patterns.
Hence, the presented definition of overlaps is a conservative extension
of the corresponding definition for PRSs. If
$\langle \ell_1 \R r_1,p,q,(\vec{x}),(\vec[l]{w}),\delta,\gamma,U,
\ell_2 \R r_2 \rangle$ is an overlap then
$\bvp{\ell_2\gamma\theta}{pq} = \vec{x},\vec[l]{w}$.
The upcoming result ensures that there exists a source term for
overlaps.

\begin{restatable}{lemma}{overlapequalsource}
\label{lem:overlapEqualSource}
If
$\langle \ell_1 \R r_1,p,q,(\vec{x}),(\vec[l]{w}),\delta,\gamma,U,
\ell_2 \R r_2 \rangle$ is an overlap then $\ell_2\gamma\theta =
\ell_2\gamma\theta[(\vec{x},\vec[l]{w}.\ell_1\delta)\theta]_{pq}$
for all $\theta \in U$.
\end{restatable}

Next, we transform overlaps into lifted critical peaks and pairs.
Note that the set of overlaps is always infinite because
$\vec[l]{w}$
is not fixed by the rewrite system. For confluence analysis, we
will only need critical pairs stemming from overlaps with $l = 0$.

\begin{definition}
Let
$\langle \ell_1 \R r_1,p,q,(\vec{x}),(\vec[l]{w}),\delta,\gamma,U,
\ell_2 \R r_2 \rangle$
be an overlap of a \textup{DPRS} $\xR$ and $\theta \in U$. By
\lemref{overlapEqualSource},
$\ell_2\gamma\theta[(\vec{x},\vec[l]{w}.\ell_1\delta)\theta]_{pq} =
\ell_2\gamma\theta$. This term can be rewritten in two ways
\[
\ell_2\gamma\theta[(\vec{x},\vec[l]{w}.r_1\delta)\theta]_{pq}
\Lab{\xR}{pq | \ell_1 \R r_1}
\ell_2\gamma\theta[(\vec{x},\vec[l]{w}.\ell_1\delta)\theta]_{pq} =
\ell_2\gamma\theta \Rab{\xR}{\epsilon | \ell_2 \R r} r_2\gamma\theta
\]
which gives rise to a \emph{lifted critical peak} of $\xR$. The
equation
$\ell_2\gamma\theta[(\vec{x},\vec[l]{w}.r_1\delta)\theta]_{pq}
\approx r_2\gamma\theta$
is called a \emph{lifted critical pair} of $\xR$.
A lifted critical peak (lifted critical pair) with $l = 0$ is just
called a \emph{critical peak} (\emph{critical pair}).
The set of (lifted) critical pairs is
denoted by $\cp{\xR}$ $(\lcp{\xR})$.
\end{definition}

\subsection{Examples}

We proceed by computing the critical pairs of the DPRSs
introduced in \secref{dhp}. All critical pair computations
can be checked with our prototype implementation.

\begin{example}
\label{exa:roundcps}
The DPRS $\xR$ from \exaref{round} has only one critical pair
\begin{align*}
\m{app}(z.H(\m{id}(z),\m{id}(z),\m{fromInt}(z)),X)
&\approx H(\m{round}(X),\m{round}(X),X)
\end{align*}
stemming from the variable overlap between the two rules.
Since
\begin{align*}
\m{app}(z.H(\m{id}(z),\m{id}(z),\m{fromInt}(z)),X)
\RbR H(\m{round}(X),\m{round}(X),X)
\end{align*}
the critical pair is joinable.
Note that it is important
for the second rule to return $\m{id}(y)$ instead of just $y$ in order to
make the critical pair joinable.
%In our opinion, having to wrap $\m{int}$ values into the
%$\m{id}$ function is a mild price to pay for succinctly expressing
%such a complicated one-step computation.
\end{example}

\begin{example}
\label{exa:fusioncps}
The DPRS $\xR$ from \exaref{fusion} does not exhibit critical pairs.
\end{example}

\begin{example}
\label{exa:sducps}
The DPRS $\xR$ from \exaref{sdu} has 30 critical pairs; 18 of them are
not joinable whereas the other 12 are trivially joinable. Every
non-joinable critical pair can be simplified by rewriting to a particular
instance of the equation schema
\begin{align*}
\m{case}_\iota&(Z,z_1.H_0(H_1(\m{inl}(z_1),z_1),
\m{inl}(z_1)),z_2.H_0(H_2(\m{inr}(z_2),z_2),\m{inr}(z_2))) \\
&\approx H_0(\m{case}_\kappa(Z,z_1.H_1(Z,z_1),z_2.H_2(Z,z_2)),z)
\end{align*}
where $\iota, \kappa \in \SET{\m{a},\m{b},\m{u}}$. Note that this
schematic equation shows how to make this system confluent: We need to
be able to detect potential partial applications of $\m{case}_\iota$
according to the arguments of nested free variables. Unfortunately,
the only possible orientation for a rule is from left to right,
but the left-hand side is not a DHP as it contains nested free variables
as well as anchors which are expanded subterms of each other.
Even though a confluent representation cannot be obtained in
our framework, through the definition of critical pairs for DPRSs this
interesting and general source of non-confluence can be determined
algorithmically.
\end{example}

\section{Confluence}
\label{sec:conf}

In this section, we will prove the two basic confluence results
which build upon the definition of critical pairs:
\begin{enumerate}[(i)]
\item
DPRSs $\xR$ are locally confluent if and only if
their critical pairs are joinable,
\item
orthogonal DPRSs are confluent.
\end{enumerate}

Confluence of terminating DPRSs can then be established by (i)
together with Newman's Lemma \cite{N42}.

\subsection{Local Confluence}
\label{sec:lconf}

The following result, our \emph{critical pair lemma},
is the main result regarding confluence of DPRSs. 

\begin{restatable}{lemma}{cplemma}
\label{lem:cp}
Let $\xR$ be a \textup{DPRS}. If $t \Lb{\xR} \cdot \Rb{\xR} u$ then
$t \Db{\xR} u$ or
$t \Rab{\xR}{*} \cdot \LRb{\lcp{\xR}} \cdot \Lab{\xR}{*} u$.
\end{restatable}

As expected from the definition of overlaps for DPRSs,
one has to adapt the proof for PRSs due to Mayr and Nipkow \cite{MN98}
quite considerably.
In particular, the statement differs from known critical pair
lemmata by (a) merely establishing connectedness
via lifted critical pairs (of which there are always infinitely many)
and (b) only guaranteeing this after rewriting. The necessity
of (a) was already exhibited in \exaref{varOverlap}, and the upcoming
example shows the need for (b).

\begin{example}
Consider the sort $\m{a}$ as well as the the function symbols
$\m{c} : (\m{a} \to \m{a}) \to \m{a}$,
$\m{d} : \m{a}$, $\m{f} : (\m{a},\m{a}) \to \m{a}$
and $\m{g}, \m{h} : \m{a} \to \m{a}$.
The DPRS $\xR$ consisting of the rules
\begin{align*}
\m{g}(X) &\R \m{h}(X) & \m{c}(y.Z(\m{g}(y))) &\R Z(\m{d})
\end{align*}
admits the non-joinable local peak
\[
\m{c}(y.\m{f}(\m{h}(y),\m{g}(y))) \Lab{\xR}{11}
\m{c}(y.\m{f}(\m{g}(y),\m{g}(y))) \Rab{\xR}{\epsilon} \m{f}(\m{d},\m{d})
\]
as well as the corresponding critical pair
$\m{c}(y.Z(\m{h}(y))) \approx Z(\m{d}) \in \cp{\xR}$.
In order to connect the two terms stemming from the local peak with the
critical pair, we first have to rewrite the left one:
$\m{c}(y.\m{f}(\m{h}(y),\m{g}(y))) \RbR
\m{c}(y.\m{f}(\m{h}(y),\m{h}(y))) \LRb{\cp{\xR}} \m{f}(\m{d},\m{d})$.
\end{example}

Even though \lemref{cp} speaks about
lifted critical pairs, it is enough to consider critical pairs for
confluence analysis.

\begin{restatable}{lemma}{lcpjoin}
\label{lem:lcpjoin}
Let $\xR$ be a \textup{DPRS}. If $\cp{\xR} \subseteq {\Db{\xR}}$
then $\lcp{\xR} \subseteq {\Db{\xR}}$.
\end{restatable}

\begin{theorem}
\label{thm:lconfluence}
A \textup{DPRS} is locally confluent if and only if all its
critical pairs are joinable.
\end{theorem}

\begin{proof}
Let $\xR$ be a DPRS. If all
critical pairs are joinable then
${\LRb{\lcp{\xR}}} \subseteq {\Db{\xR}}$
by \lemref{lcpjoin} and \thmref{clsubst}.
Thus,
${\Lb{\xR} \cdot \Rb{\xR}} \subseteq {\Db{\xR}}$ follows from
\lemref{cp}.
\end{proof}

From \thmref{lconfluence}, Newman's Lemma \cite{N42} and
\thmref{unifdecidable} we immediately obtain the following results.

\begin{corollary} 
\label{cor:confluence}
A terminating \textup{DPRS} is confluent if and only if all its
critical pairs are joinable.
\end{corollary}

\begin{corollary}
Local confluence of finite left-linear \textup{DPRSs} is decidable.
\end{corollary}

\begin{corollary}
Confluence of finite terminating left-linear \textup{DPRSs} is
decidable.
\end{corollary}
  
The locally confluent DPRSs introduced in this paper
(Examples~\ref{exa:round} and \ref{exa:fusion}) can be shown to
be terminating by the monotone algebra approach for HRSs put
forward by van de Pol \cite{vdP96}. For both examples,
our prototype can
automatically generate suitable polynomial interpretations
over the natural numbers. Hence, these DPRSs are also
confluent by \corref{confluence}.

\subsection{Orthogonal DPRSs}
\label{sec:ortho}

Our definition of orthogonal DPRSs corresponds to the
standard syntactic definition via left-linearity
and absence of critical pairs \cite{BN98,MN98}.

\begin{definition}
A \textup{DPRS} $\xR$ is \emph{non-ambiguous} if $\cp{\xR} = \varnothing$.
We call a \textup{DPRS} \emph{orthogonal} if it is both left-linear and
non-ambiguous.
\end{definition}

We prove confluence of orthogonal DPRSs by establishing the diamond
property ($\L \cdot \R \,\subseteq\, \R \cdot \L$) of their
multi-step relation which lies between one-step and
many-step rewriting. This general approach is often used to
obtain confluence of orthogonal rewrite systems \cite{T03,MN98}.

\begin{definition}
Let $\xR$ be an \textup{HRS}. We define the \emph{multi-step} relation
$\MSb{\xR}$ inductively as follows:
\begin{enumerate}[(i)]
\item \label{enu:ms1}
$\vec[k]{x}.h(\vec{s}) \MSb{\xR} \vec[k]{x}.h(\vec{t})$
if $\vec[k]{x}.s_i \MSb{\xR} \vec[k]{x}.t_i$ for all $1 \leq i \leq n$, and
\item \label{enu:ms2} \smallskip
$\vec[k]{x}.\ell\theta \MSb{\xR} \vec[k]{x}.r\delta$
if $\ell \R r \in \xR$, $\dom{\theta} \subseteq \fv{\ell}$ and
$\theta \MSb{\xR} \delta$.
\end{enumerate}
\end{definition}

This inductive definition of multi-steps is based on \cite{HM11}.
For DPRSs, it can be implemented recursively as $\theta(X)$ is
contained in $\ell\theta$ for all $X \in \dom{\theta}$. Note that this
is not the case for HRSs with nested free variables in left-hand sides.
Even though Mayr and Nipkow use a different presentation of multi-steps,
the following result can be obtained by adapting their proof
accordingly \cite{MN98}.

\begin{restatable}{lemma}{msincl}
\label{lem:msIncl}
For every \textup{HRS} $\xR$,
${\Rb{\xR}} \subseteq {\MSb{\xR}} \subseteq {\Rab{\xR}{*}}$.
\end{restatable}

For orthogonal DPRSs $\xR$, we will obtain the diamond property for
$\MSb{\xR}$ via the \emph{triangle property} \cite{T03} of $\MSb{\xR}$
(\lemref{multistepTriangle})
with respect to the following restriction of $\MSb{\xR}$.

\begin{definition}
Let $\xR$ be an \textup{HRS}. We define the
\emph{maximal multi-step} relation $\MMSb{\xR}$ inductively as
follows:
\begin{enumerate}[(i)]
\item \label{enu:mms1}
$\vec[k]{x}.h(\vec{s}) \MMSb{\xR} \vec[k]{x}.h(\vec{t})$
if $\vec[k]{x}.s_i \MMSb{\xR} \vec[k]{x}.t_i$ for all $1 \leq i \leq n$ and
$h(\vec{s})$ is not a redex, and
\item \label{enu:mms2}
$\vec[k]{x}.\ell\theta \MMSb{\xR} \vec[k]{x}.r\delta$
if $\ell \R r \in \xR$, $\dom{\theta} \subseteq \fv{\ell}$ and
$\theta \MMSb{\xR} \delta$.
\end{enumerate}
\end{definition}

Note that the relation $\MMSb{\xR}$ is deterministic if $\xR$ is an
orthogonal DPRS:
Since matching is deterministic and there are no root
overlaps, the rule $\ell \R r$ as well as the substitution $\theta$ in
case \enuref{mms2} are uniquely determined and the substitution $\delta$
is uniquely determined by induction.
Using our definition of critical pairs, we establish
\emph{coherence} of orthogonal DPRSs.

\begin{restatable}{lemma}{lhssubt}
\label{lem:lhsSubt}
Let $\xR$ be an orthogonal \textup{DPRS}, $\ell \R r \in \xR$ and
$\ell \subt s$. If $s\theta \MSb{\xR} t$ then there exists a substitution
$\delta$ such that $s\delta = t$ and $\theta|_{\fv{s}} \MSb{\xR} \delta$.
\end{restatable}

We can now just replay the proof of the triangle property for PRSs
adapted to our presentation of multi-steps \cite{MN98}.

\begin{restatable}{lemma}{multisteptriangle}
\label{lem:multistepTriangle}
For every orthogonal \textup{DPRS} $\xR$, if $s \MSb{\xR} t$ then
$t \MSb{\xR} u$ for the term $u$ that satisfies $s \MMSb{\xR} u$.
\end{restatable}

\begin{theorem}
\label{thm:orthoConfluent}
Orthogonal \textup{DPRSs} are confluent.
\end{theorem}

\begin{proof}
Let $\xR$ be an orthogonal DPRS.
By \lemref{multistepTriangle}, $\MSb{\xR}$ has the triangle property
and therefore also the diamond property \cite{T03}. Since we also have 
${\R} \subseteq {\MSb{\xR}} \subseteq {\Ra{*}}$
(\lemref{msIncl}), we conclude confluence of $\xR$ by standard
results for abstract rewriting, e.g.~\cite[Corollary 2.7.7]{BN98}.
\end{proof}

\begin{example}
Recall the DPRS $\xR$ presented in \exaref{fusion} and consider
the additional function symbol $\m{from} : \m{a} \to \m{b}$.
We define the DPRS $\xR'$ as the extension of $\xR$ by the
nonterminating rule
$\m{from}(X) \R \m{cons}(X,\m{from}(s(X)))$
which can generate infinite lists $[i,i+1,\dots]$.
Clearly, $\xR'$ is not terminating. However, it is left-linear and
non-ambiguous, so $\xR'$ is orthogonal and therefore confluent by
\thmref{orthoConfluent}.
\end{example}

\section{Case Study: Program Transformations}
\label{sec:progtrans}

It is not a coincidence that we were able to capture an interesting
program transformation as a DPRS in \exaref{fusion}.  In fact, as we
will see in this section, DPRSs can express a large class of program
transformations for functional programs. In particular, this
means that we can
use the results from the previous section
to analyze confluence of such program transformations.
Confluent program transformations
are well-behaved in the sense that there do not exist different
``optimized'' versions of a program which cannot be transformed into
each other.

Since we want to express program transformations purely by rewriting,
the inductive reasoning common to recursive definitions
has to be carried out through \emph{fusion transformation rules}
\cite{dMS99}. For each datatype, there is a fold function as well
as a corresponding fusion rule. In the following,
we will use Haskell-like notation to speak about
concepts from functional programming. For example, the
fusion rule for $\m{foldr}$ on lists is defined as follows:
\[
f\:(\m{foldr}\:({\oplus})\:e\:xs) \,=\,
\m{foldr}\:({\otimes})\:(f\:e)\:xs \quad \text{if} \quad
\lambda a,z.f\:(a \oplus z) \,=\, \lambda a,z.a \otimes f\:z
\]
Here, everything except for $\m{foldr}$ is a variable. In particular,
the binary operators $\oplus$ and $\otimes$ are higher-order variables
which can be instantiated by concrete combining functions of folds.
The general
fusion rule cannot be used in DPRSs, but each specific instance
of $f$ turns $f\:(\m{foldr}\:({\oplus})\:e\:xs)$ and
$\lambda a,z.a \otimes f\:z$ into DHPs. Hence, we can define two
DPRS rules
\begin{align*}
f\:(\m{foldr}\:({\oplus})\:e\:xs) &\R
\m{check}\:(\lambda a,z.f\:(a \oplus z))\:e\:xs\\
\m{check}\:(\lambda a,z.a \otimes f\:z)\:e\:xs &\R
\m{foldr}\:({\otimes})\:(f\:e)\:xs
\end{align*}
by unraveling the conditional equation.%
\footnote{For a completely faithful encoding, one would have to
introduce conditional DPRSs.}
This way, we can
rewrite concrete instances of the left-hand side
of the condition such that the right-hand side matches the
resulting term. Using this match, we are then able construct the
fused fold. Note that fusion transformations are
also applicable to programs written using explicit
recursion instead of folds by using the fact that the identity function
on any type is itself an instance of the fold function on that type.
For lists we obtain $\m{id} = \m{foldr}\:(:)\:[]$.

We now present two different examples which demonstrate the
usability of (confluence analysis of) DPRSs for program
transformations with fusion rules.
The program transformations modeled in the examples are taken from
\cite[Section 4]{dMS99}.
In both examples, we want to transform the naive recursive definition of
a function into an efficient definition using a fold by rewriting.
To that end, we use a rewrite system consisting of the recursive
definition of the given function $f$, a definition of the fast version
in terms of $f$, the corresponding fusion transformation
rules and additional rewrite rules which represent the required facts
to derive the fast version.

\begin{table}[t]
\centering
\caption{Signature for the DPRS $\xR$ in \exaref{fastrev}.}
\label{tab:fastrevfs}  
\begin{tabular}{c@{\quad}c@{\qquad}c@{\quad}c}
symbol & type & symbol & type \\ \midrule
$\m{nil}$ & $\m{b}$ &
$\m{foldr}_1$ & $((\m{a}, \m{b}) \to \m{b}, \m{b}, \m{b}) \to \m{b}$ \\
$\m{cons}$ & $(\m{a}, \m{b}) \to \m{b}$ &
$\m{foldr}_2$ & $((\m{a}, \m{b} \to \m{b}, \m{b}) \to \m{b},
\m{b} \to \m{b}, \m{b}, \m{b}) \to \m{b}$ \\
$\app$ & $(\m{b},\m{b}) \to \m{b}$ &
$\m{checkRev}$ & $((\m{a},\m{b}) \to \m{b},\m{b},\m{b}) \to \m{b}$ \\
$\m{sing}$ & $\m{a} \to \m{b}$ &
$\m{checkApp}$ &
$((\m{a},\m{b},\m{b}) \to \m{b}),\m{b},\m{b},\m{b}) \to \m{b}$ \\
$\m{rev}$ & $\m{b} \to \m{b}$ &
$\m{fastRev}$ & $(\m{b},\m{b}) \to \m{b}$
\end{tabular}
\end{table}

\begin{example}
\label{exa:fastrev}
Consider the sorts $\m{a}$, $\m{b}$ as well as the function
symbols given in \tabref{fastrevfs}.
The DPRS $\xR$ consisting of the following rewrite rules
\begin{align*}
\m{nil} \app Y &\R Y &
\m{sing}(X) &\R \m{cons}(X,\m{nil}) \\
\m{cons}(X,Y) \app Z &\R \m{cons}(X,Y \app Z) &
\m{rev}(\m{nil}) &\R \m{nil} \\
(Y \app Z) \app W &\R Y \app (Z \app W) &
\m{rev}(\m{cons}(X,Y)) &\R \m{rev}(Y) \app \m{sing}(X)
\end{align*} \vspace{-5ex}
\begin{align*}
&\m{rev}(\m{foldr}_1(z_1,z_2.G(z_1,z_2),E,Y)) \R
\m{checkRev}(z_1,z_2.\m{rev}(G(z_1,z_2)),E,Y) \\
&\m{checkRev}(z_1,z_2.H(z_1,\m{rev}(z_2)),E,Y) \R
\m{foldr}_1(z_1,z_2.H(z_1,z_2),\m{rev}(E),Y) \\
&\m{fastRev}(Y,Z) \R
\m{rev}(\m{foldr}_1(z_1,z_2.\m{cons}(z_1,z_2),\m{nil},Y)) \app Z  
\end{align*} \vspace{-5ex}
\begin{align*}
&\m{foldr}_1(z_1,z_2.G(z_1,z_2),E,Y) \app Z \\
&\quad \R \m{checkApp}(z_1,z_2,z_3.G(z_1,z_2) \app z_3,E,Y,Z) \\
&\m{checkApp}(z_1,z_2,z_3.F(z_1,z_4.z_2 \app z_4,z_3),E,Y,Z) \\
&\quad \R \m{foldr_2}(z_1,z_2,z_3.F(z_1,z_4.z_2(z_4),z_3),z.E \app z,Y,Z)
\end{align*}
contains a naive recursive definition of the reverse function on lists.
As can be seen by inspecting the rule for $\m{fastRev}$, the idea is
to use an accumulation parameter in order to avoid using $\app$.
In this example, we need two fusion transformations,
one transforming $\m{rev}$ into a fold and another one for $\app$
which is intended to transform the definition of $\m{fastRev}$ in terms
of the slow $\m{rev}$ function into the desired efficient fold expression:
{\allowdisplaybreaks
\begin{align*}
\m{fastRev}(Y,Z)
&\RbR \m{rev}(\m{foldr}_1(z_1,z_2.\m{cons}(z_1,z_2),\m{nil},Y)) \app Z
\\
&\RbR \m{checkRev}(z_1,z_2.\m{rev}(\m{cons}(z_1,z_2)),\m{nil},Y) \app Z
\\
&\RbR
\m{checkRev}(z_1,z_2.\m{rev}(z_2) \app \m{sing}(z_1),\m{nil},Y) \app Z \\
&\RbR
\m{foldr}_1(z_1,z_2.z_2 \app \m{sing}(z_1),\m{rev}(\m{nil}),Y) \app Z \\
&\RbR
  \m{foldr}_1(z_1,z_2.z_2 \app \m{sing}(z_1),\m{nil},Y) \app Z \\
&\RbR
\m{checkApp}(z_1,z_2,z_3.(z_2 \app \m{sing}(z_1)) \app z_3,\m{nil},Y,Z)
\\
&\RbR
\m{checkApp}(z_1,z_2,z_3.z_2 \app (\m{sing}(z_1) \app z_3),\m{nil},Y,Z)
\\
&\RbR
\m{foldr}_2(z_1,z_2,z_3.z_2(\m{sing}(z_1) \app z_3),z.\m{nil} \app z,Y,Z)
\\
&\Rab{\xR}{+} \m{foldr}_2(z_1,z_2,z_3.z_2(\m{cons}(z_1,z_3)),z.z,Y,Z)
\end{align*}}%
For this computation to succeed, it is important to apply associativity
of $\app$ before using the rule defining $\m{checkApp}$. In fact, $\xR$
is not locally confluent. Our prototype implementation reports two
critical pairs; one of these corresponds to the problem at hand. These
critical pairs are neither joinable nor simplify to an equation which can
be oriented into a DPRS rule. Even though this means that there is no
confluent DPRS representing this program transformation,
the non-joinable critical pairs which are automatically computed by our
prototype yield separation constraints for a stratification of $\RbR$
into a rewrite process which yields unique normal forms.
While we get unique normal forms from any such stratification,
it is important to choose the one which yields the desired normal
form: If we chose to exhaustively apply the rule defining $\m{checkApp}$
before the associativity rule of $\app$, we would obtain the unique but
undesired normal form
$\m{foldr}_2(z_1,z_2,z_3.\m{app}(z_2(\m{cons}(z_1,\m{nil})),z_3),z.z,Y,Z)$.
Hence, in order to make this decision, one needs more information than
the one inherent in the non-joinable critical pairs.
\end{example}

In our second example, we use natural numbers instead of lists.
Hence, instead of $\m{foldr}$ we use the fold function on natural numbers
denoted by $\m{foldn}$. Moreover, the intended program transformation
uses tupling instead of accumulation parameters. We define a specific
instance of the fold fusion rule for the function
$\m{split}\:f\:g\:a \,=\, (f\:a,g\:a)$ \cite{dMS99}:
\[
\m{split}\:f\:g\:(\m{foldn}\:s\:z\:x) \,=\,
\m{foldn}\:(\m{uncurry}\:h)\:(\m{split}\:f\:g\:z)\:x
\]
if
$\lambda y.\m{split}\:f\:g\:(s\:y)
\,=\, \lambda y.h\:(f\:y)\:(g\:y)$.
We follow the same general procedure as for the original fusion rule;
this time we have to ensure that $f$ and $g$ are instantiated such
that $\m{split}\:f\:g\:(\m{foldn}\:s\:z\:x)$ and
$\lambda y.h\:(f\:y)\:(g\:y)$ are DHPs.
\begin{table}[t]
\centering
\caption{Signature for the DPRS $\xR$ in \exaref{fastfib}.}
\label{tab:fastfibfs}  
\begin{tabular}{c@{\quad}c@{\qquad}c@{\quad}c}
symbol & type & symbol & type \\ \midrule
$\m{0}$ & $\m{a}$ &
$\m{foldn}_1$ & $(\m{a} \to \m{a}, \m{a}, \m{a}) \to \m{a}$ \\
$\m{s}$, $\m{fib}$ & $\m{a} \to \m{a}$ &
$\m{foldn}_2$ & $(\m{b} \to \m{b}, \m{b}, \m{a}) \to \m{b}$ \\
${+}$ & $(\m{a},\m{a}) \to \m{a}$ &
$\m{check}$ & $(\m{a} \to \m{b},\m{b},\m{a}) \to \m{b}$ \\
$\m{fastFib}$ & $\m{a} \to \m{b}$ &
$\m{uncurry}$ & $((\m{a},\m{a}) \to \m{b}, \m{b}) \to \m{b}$ \\
$\m{pair}$ & $(\m{a},\m{a}) \to \m{b}$ &
$\m{split}$ & $(\m{a} \to \m{a},\m{a} \to \m{a}, \m{a}) \to \m{b}$
\end{tabular}
\end{table}

\begin{example}
\label{exa:fastfib}
Consider the sorts $\m{a}$, $\m{b}$ as well as the function
symbols given in \tabref{fastfibfs}.
The DPRS $\xR$ consisting of the following rewrite rules
\begin{align*}
&\m{fib}(\m{0}) \R \m{0} \qquad
\m{fib}(\m{s}(\m{0})) \R \m{s}(\m{0}) \qquad
\m{fib}(\m{s}(\m{s}(X))) \R \m{fib}(\m{s}(X)) + \m{fib}(X) \\  
&\m{split}(z.F(z),z.G(z),X) \R \m{pair}(F(X),G(X)) \\
&\m{fastFib}(X) \R
\m{split}(z.\m{fib}(z),z.\m{fib}(\m{s}(z)),\m{foldn}_1(z.\m{s}(z),\m{0},X))
\\
&\m{split}(z.F(z),z.G(z),\m{foldn}_1(z.S(z),Z,X)) \\
&\quad \R \m{check}(z_1.\m{split}(z_2.F(z_2),z_2.G(z_2),z_1),
\m{split}(z.F(z),z.G(z),Z),X) \\
&\m{check}(z.H(\m{fib}(z),\m{fib}(\m{s}(z))),Z,X) \R
\m{foldn}_2(z_1.\m{uncurry}(z_2,z_3.H(z_2,z_3),z_1),Z,X)
\end{align*}
contains a naive recursive definition of a function computing the
Fibonacci numbers. As can be seen by inspecting the rule for
$\m{fastFib}$, the idea is to compute tuples of successive Fibonacci
numbers simultaneously instead of computing them one at a time. Using
$\xR$ we can compute an efficient version of $\m{fastFib}$:
{\allowdisplaybreaks\begin{align*}
&\mr{\phantom{\RbR}} \m{fastFib}(X) \\
&\RbR
\m{split}(z.\m{fib}(z),z.\m{fib}(\m{s}(z)),\m{foldn}_1(z.\m{s}(z),\m{0},X))
\\
&\RbR \m{check}(z_1.\m{split}(z_2.\m{fib}(z_2),z_2.\m{fib}(\m{s}(z_2)),
\m{s}(z_1)),\m{split}(z.\m{fib}(z),z.\m{fib}(\m{s}(z)),\m{0},X)) \\
&\Rab{\xR}{+}
\m{check}(z_1.\m{pair}(\m{fib}(\m{s}(z_1)),\m{fib}(\m{s}(\m{s}(z_1)))),
\m{pair}(\m{fib}(\m{0}),\m{fib}(\m{s}(\m{0}))),X) \\
&\Rab{\xR}{+}
\m{check}(z_1.\m{pair}(\m{fib}(\m{s}(z_1)),\m{fib}(\m{s}(\m{s}(z_1)))),
\m{pair}(\m{0},\m{s}(\m{0})),X) \\
&\RbR
\m{check}(z_1.\m{pair}(\m{fib}(\m{s}(z_1)),
\m{fib}(\m{s}(z_1)) + \m{fib}(z_1)),\m{pair}(\m{0},\m{s}(\m{0})),X) \\
&\RbR
\m{foldn}_2(z_1.\m{uncurry}(z_2,z_3.\m{pair}(z_3,z_3 + z_2),z_1),
\m{pair}(\m{0},\m{s}(\m{0})),X)
\end{align*}}%
As before, our prototype implementation reports two non-joinable critical
pairs, so $\xR$ is not locally confluent. This time, it is necessary to
apply the $\m{split}$ fusion transformation rule before the rule defining
$\m{split}$ in order to get the desired normal form.
\end{example}

Fusion transformation rules are prominent in the rewrite rules of
Haskell's GHC compiler \cite{PJTH01}. Recently, the rewrite rules of
GHC have been extended to Miller's patterns \cite{R24}. The results
given in the present paper serve as a perfect starting point for a further
extension to DHPs. In particular, the non-joinable critical pairs
automatically given by our prototype implementation yield separation
constraints for stratifying the rewrite relation into a phased rewrite
process with unique normal forms. At this point, the concrete order of
the stratification which yields the desired normal form has to be given
manually. In future work, our automatic confluence analysis could be
augmented by additional static analysis techniques in order to
(partially) automate the whole process.

\section{Conclusion}
\label{sec:conclusion}

By giving a complete definition of overlaps at variable positions with
at most one critical peak per unifier in the minimal complete set of
unifiers, we lifted the notion of higher-order critical pairs from PRSs
to DPRSs. DPRSs do not only include PRSs but also allow for succinct
representations of complex higher-order computations including
program transformations from functional programming.

For finite terminating PRSs, confluence is decidable because we
only have to check joinability of finitely many critical pairs.
We were able to show that this property at least carries over to
left-linear DPRSs. The computation
of critical pairs and the corresponding joinability checks for
all examples in this paper can be done automatically by our
prototype implementation which is not limited to left-linear DPRSs
as demonstrated by its applicability to \exaref{sdu}.

Beyond the two basic confluence results established in
\secref{conf}, our definition of critical pairs lays the
foundation for adapting many more confluence methods for
non-terminating DPRSs such as
development-closedness \cite{vO97} and even methods for
non-left-linear systems \cite{FJ25}. All these confluence criteria
demand special joinability conditions on (variations of) critical pairs.
%\exaref{repl} also hints at
%the possibility of a higher-order completion procedure for DPRSs.
%Finally, it is important to note that DPRSs are often inherently
%non-confluent when nested free variables in critical pairs
%stemming from variable overlaps cannot be removed by
%rewriting. In this case, the simplified critical pair
%cannot be oriented into a DPRS rule.
%Hence, as we have already seen
%in \exaref{sducps}, there is a need for an even more permissive class
%of higher-order rewrite systems in order to ensure confluence of
%DPRSs. However, going in this direction completely changes the nature
%of our model of computation: For such rewrite systems, matching is not
%deterministic anymore; a given rule can now be applied in
%many different ways at a single position. Consequently,
%the number of possible rewrite sequences starting from a term and therefore
%also the number of critical pairs explodes as illustrated in \exaref{hrs}.
%For now, as shown in
%\secref{progtrans}, the main strength of confluence analysis for DPRSs
%lies in the automatic generation of negative answers which
%can be used to stratify DPRS computations to yield the desired unique
%normal forms.

\bibliography{references}

\appendix

\section{Proof of \thmref{unifdecidable}}

In this section, we assume familiarity with and use notation from
\cite{NM26a}. In the following, the \emph{size} $|s|$ of
a term is defined inductively as
$|\vec[k]{x}.h(\vec{t})| = 1 + \sum_{1 \leq i \leq n} t_i$.
Similarly, the \emph{height} $\m{h}(s)$ of a term is defined inductively
as $\m{h}(\vec[k]{x}.h(\vec{t})) =
1 + \m{max}\,\SET{\m{h}(t_i) \mid 1 \leq i \leq n}$.

\begin{definition}
Let $E$ be a set of unordered equations. Its free variables are
defined as
$\fv{E} = \bigcup \SET{\fv{s} \cup \fv{t} \mid s \approx t \in E}$.
Furthermore, $E$ is \emph{linear} if each $X \in \fv{E}$ occurs only once
in $E$. For linear $E$ and $X \in \fv{E}$ we define the
\emph{cost function} $\m{c} : \fv{E} \to \mathbb{N} \cup \SET{\infty}$ as
follows:
\[
\m{c}(X) = \begin{cases}
\m{h}(t) &\text{if $s \approx t \in E$ where $X$ is the head of $s$} \\
\m{h}(t_i) &\text{if $\vec[k]{x}.h(\vec{s}) \approx \vec[k]{x}.h(\vec{t})
\in E$ where $X$ is the head of $s_i$} \\
\infty &\text{otherwise}                        
\end{cases}
\]
\end{definition}

\begin{lemma}
\label{lem:lininvariant}
If $(E,\theta)$ is a unification pair with linear $E$
and $(E,\theta) \ustep (E',\theta')$ then $E'$ is also linear.
\end{lemma}

\begin{proof}
The result is trivial for \urem and \udec. For \uvar and \uffne,
$E'$ is linear because $E' \subseteq E$ and $E$ is linear by assumption.
Note that \uffe is not applicable as $E$ is linear.
Finally, consider \uimt and \uprj. By the linearity assumption, only one
equation is altered. Hence, the result follows as
$u\SEt{\vec{x} \mapsto \vec{s}}$ is a linear term which only contains
fresh free variables.
\end{proof}

\begin{lemma}
\label{lem:unifterm}
The derivation relation $\ustep$ terminates for linear finite sets of
unordered equations.
\end{lemma}

\begin{proof}
Given a unification pair $(E,\theta)$ with linear $E$,
the measure $\m{m}(E,\theta) = (M,n)$
with multiset $M = \SET{\m{c}(X) \mid X \in \fv{E}}$
and $n = \sum\SET{|s| + |t| \mid s \approx t \in E}$ is well-defined.
By \lemref{lininvariant}, this is the case for all unification pairs
which may come up during our analysis. Let $>$ be the standard
order on natural numbers and $\succ$ be the
well-founded order on $\mathbb{N} \cup \SET{\infty}$ which extends $>$
with the case $\infty \succ n$
for all $n \in \mathbb{N}$. We use $\geq$ and $\succeq$ to denote the
reflexive closures of $>$ and $\succ$, respectively.
For an arbitrary derivation step $(E,\theta) \ustep (E',\theta')$, we
show $\m{m}(E,\theta) = (M,n) \mathrel{(\succ_\mul,>)_\lex} (M',n') =
\m{m}(E',\theta')$: For \urem and \udec, we have $M \succeq_\mul M'$ and
$n > n'$. For \uvar, \uimt, \uprj and \uffne, we have $M \succ_\mul M'$.
As before, \uffe is not applicable because $E$ is linear.
\end{proof}

\unifdecidable*

\begin{proof}
By definition, $E = \SET{s \approx t}$ is a linear set of unordered
equations. Hence, \lemref{unifterm} yields that there are no
infinite derivations with $\ustep$. Since $\ustep$ is also finitely
branching, $\mcsus{E}$ is finite and computable.
\end{proof}

\section{Proof of \lemref{overlapEqualSource}}

\begin{definition}
Given $p \in \pos{s}$, we write $p \in \posGEfv{s}$ if there exists a
position $q \in \pos{s}$ with $q \leq p$ such that the head of $s|_q$
is a free variable. Otherwise, $p \in \posLTfv{s}$.
\end{definition}

\overlapequalsource*

\begin{proof}
Let $\theta \in U$. It is sufficient to prove that
$(\vec{x},\vec[l]{w}.\ell_1\delta)\theta = \ell_2\gamma\theta|_{pq}$.
Note that
$(\vec{x},\vec[l]{w}.\ell_1\delta)\theta = \ell_2\gamma|_{pq}\theta$
as $\theta$ is a unifier of the two terms. Hence, we only have to show
$\ell_2\gamma|_{pq}\theta = \ell_2\gamma\theta|_{pq}$. If
$pq \in \posLTfv{\ell_2\gamma}$ then the result follows
immediately. Otherwise, $pq \in \posGEfv{\ell_2\gamma}$. Since $U$
is a subset of a minimal complete set of unifiers,
$\dom{\theta} \subseteq \fv{\vec{x},\vec[l]{w}.\ell_1\delta} \cup
\fv{\ell_2\gamma|_{pq}}$. Moreover, any free
variable occurring in $\ell_2\gamma$ does not occur below
itself because $\ell_2$ is a DHP. Hence, whenever the head of
$\ell_2\gamma|_r$ for $r < pq$ is a free variable $Y$ then
$Y \notin \dom{\theta}$. Again, we obtain
$\ell_2\gamma|_{pq}\theta = \ell_2\gamma\theta|_{pq}$
as desired.
\end{proof}

\section{Proof of \lemref{cp}}

\tikzstyle{term}=[regular polygon,regular polygon sides=3,is=0,os=5]
\tikzstyle{context}=[term,fill=gray!15,draw=gray!40]
\tikzstyle{l1}=[term,fill=green!70!black!30,draw=green!60!black!40]
\tikzstyle{r1}=[term,fill=blue!70!black!30,draw=blue!60!black!40]
\tikzstyle{l2}=[term,fill=green!50!black!20,draw=green!60!black!40]
\tikzstyle{r2}=[term,fill=blue!50!black!20,draw=blue!60!black!40]
\tikzstyle{l2g}=[l2,fill=green!20!gray!20]
\tikzstyle{r2g}=[r2,fill=blue!20!gray!20]
\tikzstyle{l1g}=[l1,fill=blue!20!gray!20]
\tikzstyle{c}=[context,mis=20,as]
\tikzstyle{i}=[mis=10,an]
\tikzstyle{m}=[mis=23,an]
\tikzstyle{m'}=[mis=14,an]
\tikzstyle{o}=[mis=15,as]
\tikzstyle{o'}=[mis=30,as]
\tikzstyle{o''}=[mis=40,as]
\tikzstyle{ci}=[mis=10,an]
\tikzstyle{ic}=[mis=6,as]
\tikzstyle{ic2}=[mis=10,as]
\tikzstyle{ci'}=[mis=6,an]

\begin{definition}
A set of variables $X$ is \emph{fresh for} a substitution $\theta$
if $(\dom{\theta} \cup \fv{\theta}) \cap X = \varnothing$.
\end{definition}

We start with an auxiliary lemma which we will use to connect
overlapping rewrite steps to our formal definition of overlaps. 
The statement and proof closely follows 
\cite[Lemma 4.3]{MN98} but without the restriction to patterns
which is also not needed in the proof given there.

\begin{lemma}
\label{lem:unifierExistence}
Consider the terms $s$ and $t$ as well as the position $p \in \pos{t}$.
Let $\bvp{t}{p} = \vec{x}$ and $\delta$ be an $\vec{x}$-lifter of $s$ away
from $\fv{t}$. If there exist substitutions $\theta_1$, $\theta_2$ such
that $\vec{x}.s\theta_1 = t|_p\theta_2$ then $\vec{x}.s\delta$ and $t|_p$
are unifiable.
\end{lemma}

\begin{proof}
Without loss of generality, $\dom{\theta_1} \subseteq \fv{s}$ as well as
$\dom{\theta_2} \subseteq \fv{t}$. Let $\theta_1' =
\SET{\rho(X) \mapsto \vec{x}.\theta_1(X) \mid X \in \fv{s}}$
where $\rho$ is the variable renaming corresponding to $\delta$.
Note that $\theta_1'(\delta(X)) = \theta_1(X)$ for $X \in \fv{s}$, so
$\delta\theta_1'|_{\fv{s}} = \theta_1|_{\fv{s}}$.
For $\theta_0 = \theta_1' \cup \theta_2$, the following holds:
\begin{align*}
(\vec{x}.s\delta)\theta_0
&= (\vec{x}.s\delta)\theta_1'
\tag{$\delta$ away from $\fv{t} \supseteq \dom{\theta_2}$} \\ 
&= \vec{x}.s\delta\theta_1'
\tag{$\SEt{\vec{x}}$ fresh for $\theta_1'$} \\
&= \vec{x}.s\theta_1
\tag{$\delta\theta_1'|_{\fv{s}} = \theta_1|_{\fv{s}}$} \\
&= t|_p\theta_2
\tag{assumption} \\
&= t|_p\theta_0
\tag{$\dom{\theta_1'} \cap \fv{t} = \varnothing$}
\end{align*}
Hence, $\theta_0$ is a unifier of $\vec{x}.s\delta$ and $t|_p$.
\end{proof}

Before finally proving \lemref{cp}, we recall a central theorem
about closure of rewriting under substitutions in HRSs.

\begin{theorem}[Theorem 3.9 from \cite{MN98}]
\label{thm:clsubst}  
Let $\xR$ be an \textup{HRS}. If $s \Rab{\xR}{*} t$ and
$\theta \Rab{\xR}{*} \theta'$ then $s\theta \Rab{\xR}{*} t\theta'$
\end{theorem}

\cplemma*

\begin{proof}
Let $t \Lab{\xR}{p_1 \mid \ell_1 \R r_1 \mid \theta_1} s
\Rab{\xR}{p_2 \mid \ell_2 \R r_2 \mid \theta_2} u$,
$\bvp{s}{p_1} = \vec[k]{z},\vec[o_1]{w}$ and
$\bvp{s}{p_2} = \vec[k]{z},\vec[o_2]{w}$ where $\vec[k]{z}$ is the
maximal common prefix of $\bvp{s}{p_1}$ and $\bvp{s}{p_2}$. Without
loss of generality we assume
$\dom{\theta_1} \subseteq \fv{\ell_1}$ and
$\dom{\theta_2} \subseteq \fv{\ell_2}$. If $p_1 \parallel p_2$ then
$\vec[k]{z}$ is empty and $t \Db{\xR} u$ can be established as
follows:
\[
t = s[\vec[o_1]{w}.r_1\theta_1]_{p_1}
\RbR s[\vec[o_1]{w}.r_1\theta_1]_{p_1}
[\vec[o_2]{w}.r_2\theta_2]_{p_2}
\LbR s[\vec[o_2]{w}.r_2\theta_2]_{p_2} = u
\]
Now consider the case where one position is above the other. Since
the statement we are about to prove is symmetric, we may assume
$p_1 \geq p_2$ with $p = p_1 \setminus p_2$.
Hence, $o_2 = 0$ and $\bvp{s}{p_2} = \vec[k]{z}$.
If $p = \epsilon$ and $\ell_1 \R r_1$,
$\ell_2 \R r_2$ are variants then $\ell_1\theta_1 = \ell_2\theta_2$ and
$\ell_1\mu = \ell_2$, $r_1\mu = r_2$ for some substitution $\mu$. Hence,
$\ell_1\theta_1 = \ell_1\mu\theta_2$, so
$\theta_1|_{\fv{\ell_1}} = \mu\theta_2|_{\fv{\ell_1}}$. Since
$\fv{r_1} \subseteq \fv{\ell_1}$, we obtain $r_1\theta_1 = r_1\mu\theta_2$
and therefore $r_1\theta_1 = r_2\theta_2$. Thus,
$t = s[\vec[k]{z}.r_1\theta_1]_{p_2} =
s[\vec[k]{z}.r_2\theta_2]_{p_2} = u$.
Otherwise $p \neq \epsilon$ or $\ell_1 \R r_1$ and $\ell_2 \R r_2$
are not variants.
We proceed by case analysis on whether $p \in \posLTfv{\ell_2}$ or
$p \in \posGEfv{\ell_2}$.

Assume $p \in \posLTfv{\ell_2}$ and let $\bvp{\ell_2}{p} = \vec{x}$.
Since the head of $\ell_1$ is a function symbol, $p \in \posf{\ell_2}$.
We obtain
$\vec{x}.\ell_1\theta_1 = \ell_2\theta_2|_p = \ell_2|_p\theta_2$.
Let $\delta$ be an $\vec{x}$-lifter of $\ell_1$ away from
$\fv{\ell_2}$. By \lemref{unifierExistence},
$\theta_0 = \theta_1' \cup \theta_2$ is a unifier of
$\vec{x}.\ell_1\delta$ and $\ell_2|_p$ where $\theta_1'$ is
defined as in the lemma. Hence, there is a minimal complete
set of unifiers $U \neq \varnothing$ of the two terms and
$\langle \ell_1 \R r_1,p,\epsilon,(\vec{x}),(),\delta,\varnothing,U,
\ell_2 \R r_2 \rangle$ is an overlap. Furthermore, there exists a unifier
$\theta \in U$ such that $\theta_0 = \theta\mu$ for some
substitution $\mu$. Thus, we can represent $t$ as follows:
\begin{align*}
t
&= s[\vec[k]{z}.\ell_2\theta_2[\vec{x}.r_1\theta_1]_p]_{p_2} \\
&= s[\vec[k]{z}.\ell_2\theta_2[\vec{x}.r_1\delta\theta_1']_p]_{p_2}
\tag{$\delta\theta_1'|_{\fv{r_1}} = \theta_1|_{\fv{r_1}}$} \\  
&= s[\vec[k]{z}.\ell_2\theta_2[(\vec{x}.r_1\delta)\theta_1']_p]_{p_2}
\tag{$\SEt{\vec{x}}$ is fresh for $\theta_1'$} \\
&= s[\vec[k]{z}.\ell_2\theta_{2}[(\vec{x}.r_1\delta)\theta_0]_p]_{p_2}
\tag{$\delta$ away from $\fv{\ell_2} \supseteq \dom{\theta_2}$} \\
&= s[\vec[k]{z}.\ell_2\theta_0[(\vec{x}.r_1\delta)\theta_0]_p]_{p_2}
\tag{$\dom{\theta_1'} \cap \fv{\ell_2} = \varnothing$} \\
&= s[\vec[k]{z}.\ell_2\theta\mu[(\vec{x}.r_1\delta)\theta\mu]_p]_{p_2} \\
&= s[\vec[k]{z}.(\ell_2\theta[(\vec{x}.r_1\delta)\theta]_p)\mu]_{p_2}
\tag{$p \in \posLTfv{\ell_2}$}
\end{align*}
Similarly,
\[
u = s[\vec[k]{z}.r_2\theta_2]_{p_2} = s[\vec[k]{z}.r_2\theta_0]_{p_2} =
s[\vec[k]{z}.r_2\theta\mu]_{p_2}
\]
as $\dom{\theta_1'} \cap \fv{r_2} = \varnothing$.
Thus, $t \LRb{\cp{\xR}} s[\vec[k]{z}.r_2\theta\mu]_{p_2} = u$
by an application of the critical pair
$\ell_2\theta[(\vec{x}.r_1\delta)\theta]_p \approx r_2\theta$,
as desired.

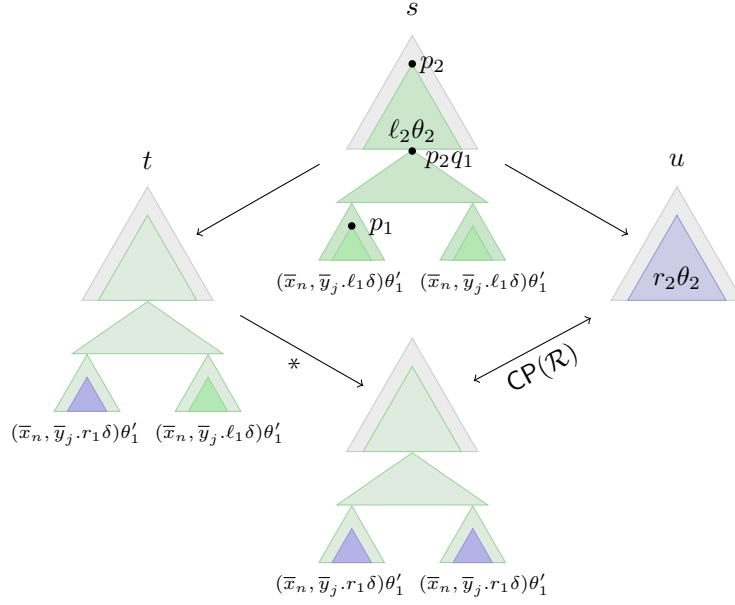
\begin{figure}[t]
\begin{center}
\begin{tikzpicture}[transform shape]
\begin{scope}
\node [c] (1) at (0,0) {};
\node [l2,o] (1o) at (0,0) {};
\node [l2,m,yscale=.4] (1m) at (1o.south) [ys=7.5] {};
\node [l2,ci] (1cil) at (1m.south) [xs=-8,ys=4.1] {};
\node [l2,ci] (1cir) at (1m.south) [xs=8,ys=4.1] {};
\node [l1,ic] (1icl) at (1cil.south) {};
\node [l1,ic] (1icr) at (1cir.south) {};
\node at (1.north) {\large $s$};
\node at (1o.north) [ys=-3.5,circle,is=1,fill] {};
\node at (1o.north) [above,xs=2.7,ys=-6] {$p_2$};
\node [as] at (1o) [above,ys=-4] {$\ell_2\theta_2$};
\node at (1m.north) [ys=-1.5,circle,is=1,fill] {};
\node at (1m.north) [above,xs=5,ys=-5] {$p_2q_1$};
\node at (1icl.north) [ys=-3.5,circle,is=1,fill] {};
\node at (1icl.north) [above,xs=4,ys=-6] {$p_1$};
\node at (1icl.south) [below,ys=1.5,xs=-1.5]
 {\scriptsize $(\vec{x},\vec[j]{y}.\ell_1\delta)\theta_1'$};
\node at (1icr.south) [below,ys=1.5,xs=1.5]
{\scriptsize $(\vec{x},\vec[j]{y}.\ell_1\delta)\theta_1'$};
\end{scope}
\begin{scope}[xs=-35,ys=-20]
\node [c] (3) at (0,0) {};
\node [l2g,o] (3o) at (0,0) {};
\node [l2g,m,yscale=.4] (3m) at (3o.south) [ys=7.5] {};
\node [l2g,ci] (3cil) at (3m.south) [xs=-8,ys=4.1] {};
\node [l2g,ci] (3cir) at (3m.south) [xs=8,ys=4.1] {};
\node [r1,ic] (3icl) at (3cil.south) {};
\node [l1,ic] (3icr) at (3cir.south) {};
\node at (3.north) {\large $t$};
\node at (3icl.south) [below,ys=1.5,xs=-1.5]
 {\scriptsize $(\vec{x},\vec[j]{y}.r_1\delta)\theta_1'$};
\node at (3icr.south) [below,ys=1.5,xs=1.5]
 {\scriptsize $(\vec{x},\vec[j]{y}.\ell_1\delta)\theta_1'$};
\draw[->,sh=2] (1) -- (3);
\end{scope}
\begin{scope}[xs=35,ys=-20]
\node [c] (2) at (0,0) {};
\node [r2,o] (2o) at (0,0) {};
\node at (2.north) {\large $u$};
\node at (2o.south) [above,ys=2] {$r_2\theta_2$};
\draw[->,sh=2] (1) -- (2);
\end{scope}
\begin{scope}[ys=-40]
\node [c] (4) at (0,0) {};
\node [l2g,o] (4o) at (0,0) {};
\node [l2g,m,yscale=.4] (4m) at (4o.south) [ys=7.5] {};
\node [l2g,ci] (4cil) at (4m.south) [xs=-8,ys=4.1] {};
\node [l2g,ci] (4cir) at (4m.south) [xs=8,ys=4.1] {};
\node [r1,ic] (4icl) at (4cil.south) {};
\node [r1,ic] (4icr) at (4cir.south) {};
\node at (4icl.south) [below,ys=1.5,xs=-1.5]
 {\scriptsize $(\vec{x},\vec[j]{y}.r_1\delta)\theta_1'$};
\node at (4icr.south) [below,ys=1.5,xs=1.5]
 {\scriptsize $(\vec{x},\vec[j]{y}.r_1\delta)\theta_1'$};
\draw[->,sh=2] (3) -- (4) node[midway,below,sloped] {$*$};
\end{scope}
\draw[<->,sh=7] (2o) -- (4) node[midway,below,sloped] {$\cp{\xR}$};
\end{tikzpicture}
\caption{Anchor overlap case in the proof of \lemref{cp}.}
\label{fig:cpCase3}
\end{center}
\end{figure}

Now assume $p \in \posGEfv{\ell_2}$. By definition, $p = q_1q_2$
such that $\ell_2|_{q_1} = \vec{x}.Y(\vec[m]{s})$.
We proceed with another case distinction.

First assume that $\ell_1$ overlaps with a subterm of one of the
anchors $\vec[m]{s}$. More formally, there is some $1 \leq i \leq m$ and
$r \in \pos{s_i}$ such that
\[
\vec{x},\vec[l]{w},y_{o'+1},\dots,y_j.\ell_1\theta_1 =
\ell_2\theta_2|_p = \vec{x},\vec[l]{w},y_{o'+1},\dots,y_o.s_i'|_r\]
where $\bvp{s_i}{r} = \vec[j]{y}$,
$\vec{x}.s_i = \vec{x},\vec[o]{y}.s_i'$ with $s_i'$ basic
and $0 \leq o' \leq o$. Since $\ell_2$ is a DHP,
$\fv{\vec{x}.s_i|_r} = \varnothing$ and we obtain
$\vec{x},\vec[j]{y}.\ell_1\theta_1 = \vec{x}.s_i|_r = \ell_2|_{q_1ir}$
as well as $q_1ir \in \posf{\ell_2}$.
Let $\delta$ be an $\vec{x},\vec[j]{y}$-lifter of
$\ell_1$ away from $\fv{\ell_2}$. By \lemref{unifierExistence},
$\theta_0 = \theta_1' \cup \theta_2$ is a unifier of
$\vec{x},\vec[l]{y}.\ell_1\delta$ and $\ell_2|_{q_1ir}$
where $\theta_1'$ is defined as in the lemma. Let $U$ be a minimal
complete set of unifiers of the two terms. Since
$\fv{\ell_2|_{q_1ir}} = \varnothing$, this unification problem
actually corresponds to a matching problem which by
\thmref{matchUnique} has the unique solution $\theta_1' \in U$. Hence,
\[
\langle \ell_1 \R r_1,q_1ir,\epsilon,(\vec{x},\vec[j]{y}),(),
\delta,\varnothing,\SET{\theta_1'},\ell_2 \R r_2 \rangle
\]
is an overlap. Furthermore,
$\theta_0 = \theta_1'\theta_2$ since $\fv{\theta_1'} = \varnothing$.
We obtain
\[
t = s[\vec[k]{z}.\ell_2\theta_2[\vec{x},\vec[l]{w},y_{o'+1},\dots,
y_j.r_1\theta_1]_p]_{p_2} \Rab{\xR}{*} s[\vec[k]{z}.
(\ell_2[\vec{x},\vec[j]{y}.r_1\theta_1]_{q_1ir})\theta_2]_{p_2}
\]
as $\fv{\vec{x},\vec[j]{y}.r_1\theta_1} = \varnothing$ and if
$\theta_2(Y) = \vec[m]{y}.v$, $v$ may contain multiple occurrences of
$y_i$. As before, we can present this reduct of $t$ as follows:
\begin{align*}
t \,\Rab{\xR}{*}&\,\, s[\vec[k]{z}.
(\ell_2[\vec{x},\vec[j]{y}.r_1\theta_1]_{q_1ir})\theta_2]_{p_2} \\
=&\,\, s[\vec[k]{z}.
(\ell_2[\vec{x},\vec[j]{y}.r_1\delta\theta_1']_{q_1ir})\theta_2]_{p_2}
\tag{$\delta\theta_1'|_{\fv{r_1}} = \theta_1|_{\fv{r_1}}$} \\  
=&\,\, s[\vec[k]{z}.
(\ell_2[(\vec{x},\vec[l]{y}.r_1\delta)\theta_1']_{q_1ir})\theta_2]_{p_2}
\tag{$\SEt{\vec{x},\vec[j]{y}}$ is fresh for $\theta_1'$} \\
=&\,\, s[\vec[k]{z}.(\ell_2\theta_1'[(\vec{x},\vec[j]{y}.
r_1\delta)\theta_1']_{q_1ir})\theta_2]_{p_2}
\tag{$\dom{\theta_1'} \cap \fv{\ell_2} = \varnothing$}
\end{align*}
Similarly, $u = s[\vec[k]{z}.r_2\theta_2]_{p_2} =
s[\vec[k]{z}.r_2\theta_1'\theta_2]_{p_2}$
as $\dom{\theta_1'} \cap \fv{r_2} = \varnothing$.
Thus, $t \Rab{\xR}{*} \cdot \LRb{\cp{\xR}}
s[\vec[k]{z}.r_2\theta_1'\theta_2]_{p_2} = u$ as desired. This case is
illustrated in \figref{cpCase3}.

Next assume there is no underlying overlap in the anchors
$\vec[m]{s}$, so $q_2 \in \pos{\theta_2(Y)}$. Let
$Y : (\vec[m]{\sigma}) \to a$, $\ell_1 : b$ and
$\bvp{\theta_2(Y)}{q_2} =
\vec[m]{y},\vec[l]{w} : \vec[m]{\sigma},\vec[l]{\tau}$.
Consider the fresh
variables $Y' : (\vec[m]{\sigma},\vec[l]{\tau}) \to b$,
$Y'' : ((\vec[l]{\tau}) \to b,\vec[m]{\sigma}) \to a$
and define the substitution
$\gamma = \SET{Y \mapsto \vec[m]{y}.Y''(\vec[l]{w}.
Y'(\vec[m]{y}\cf,\vec[l]{w}\cf),\vec[m]{y}\cf)}$.
From the substitution
\[
\theta_2' = \theta_2 \cup {}
\SEt{Y' \mapsto \vec[m]{y}.Y(\vec[m]{y}\cf)\theta_2|_{q_2},
Y'' \mapsto z,\vec[m]{y}.y(\vec[m]{y}\cf)\theta_2[\vec[l]{w}.
z(\vec[l]{w})]_{q_2}}
\]
we obtain $\theta_2|_{\fv{\ell_2}} = \gamma\theta_2'|_{\fv{\ell_2}}$
and $\ell_2\gamma|_{q_11}\theta_2' = \ell_2\theta_2|_{q_1q_2} =
\vec{x},\vec[l]{w}.\ell_1\theta_1$.
Let $\delta$ be an $\vec{x},\vec[l]{w}$-lifter of $\ell_1$
away from $\fv{\ell_2\gamma}$. By \lemref{unifierExistence},
$\theta_0 = \theta_1' \cup \theta_2'$ is a unifier of
$\vec[n]{x}.\ell_1\delta$ and $\ell_2\gamma|_{q_11}$ where $\theta_1''$ is
defined as in the lemma. Let $U'$ be a minimal complete set of unifiers of
the two terms, so there exists a unifier $\theta \in U'$ such that
$\theta_0 = \theta\mu$ for some substitution $\mu$. Hence, we
can represent $t$ as follows:
\begin{align*}
t
&= s[\vec[k]{z}.\ell_2\theta_2[\vec{x},\vec[l]{w}.
r_1\theta_1]_p]_{p_2} \\  
&= s[\vec[k]{z}.\ell_2\gamma\theta_2'[\vec{x},\vec[l]{w}.
r_1\theta_1]_p]_{p_2}
\tag{$\theta_2|_{\fv{\ell_2}} = \gamma\theta_2'|_{\fv{\ell_2}}$} \\
&= s[\vec[k]{z}.\ell_2\gamma\theta_2'[\vec{x},\vec[l]{w}.
r_1\delta\theta_1']_p]_{p_2}
\tag{$\delta\theta_1'|_{\fv{r_1}} = \theta_1|_{\fv{r_1}}$} \\
&= s[\vec[k]{z}.\ell_2\gamma\theta_2'[(\vec{x},\vec[l]{w}.
r_1\delta)\theta_1']_p]_{p_2}
\tag{$\SEt{\vec{x},\vec[l]{w}}$ is fresh for $\theta_1'$} \\
&= s[\vec[k]{z}.\ell_2\gamma\theta_2'[(\vec{x},\vec[l]{w}.
r_1\delta)\theta_0]_p]_{p_2}
\tag{$\delta$ away from
$\fv{\ell_2\gamma} \subseteq \dom{\theta_2'}$} \\
&= s[\vec[k]{z}.\ell_2\gamma\theta_0[(\vec{x},\vec[l]{w}.
r_1\delta)\theta_0]_p]_{p_2}
\tag{$\dom{\theta_1'} \cap \fv{\ell_2\gamma} = \varnothing$} \\
&= s[\vec[k]{z}.\ell_2\gamma\theta\mu[(\vec{x},\vec[l]{w}.
r_1\delta)\theta\mu]_p]_{p_2} \\
&= s[\vec[k]{z}.(\ell_2\gamma\theta[(\vec{x},\vec[l]{w}.
r_1\delta)\theta]_{q_11})\mu]_{p_2}
\tag{$Y'' \notin \dom{\theta}$}
\end{align*}
Similarly,
\[
u
= s[\vec[k]{z}.r_2\theta_2]_{p_2}
= s[\vec[k]{z}.r_2\gamma\theta_2']_{p_2}
= s[\vec[k]{z}.r_2\gamma\theta_0]_{p_2}
= s[\vec[k]{z}.r_2\gamma\theta\mu]_{p_2}
\]
because $\theta_2|_{\fv{\ell_2}} = \gamma\theta_2'|_{\fv{\ell_2}}$ and
$\dom{\theta_1'} \cap \fv{r_2\gamma} = \varnothing$.
Let $U$ be defined from $U'$ as in \defref{overlap}\enuref{overlap-iv}(b).
By our assumption that there is no underlying overlap in the arguments of
$Y$, $\theta(Y') \neq \vec[m]{y},\vec[l]{w}.y_i\cf$ for some
$1 \leq i \leq m$. We proceed by a final
case analysis on whether
$\theta(Y') = \vec[m]{y},\vec[l]{w}.\ell_1\mu_1$ for some $\mu_1$.

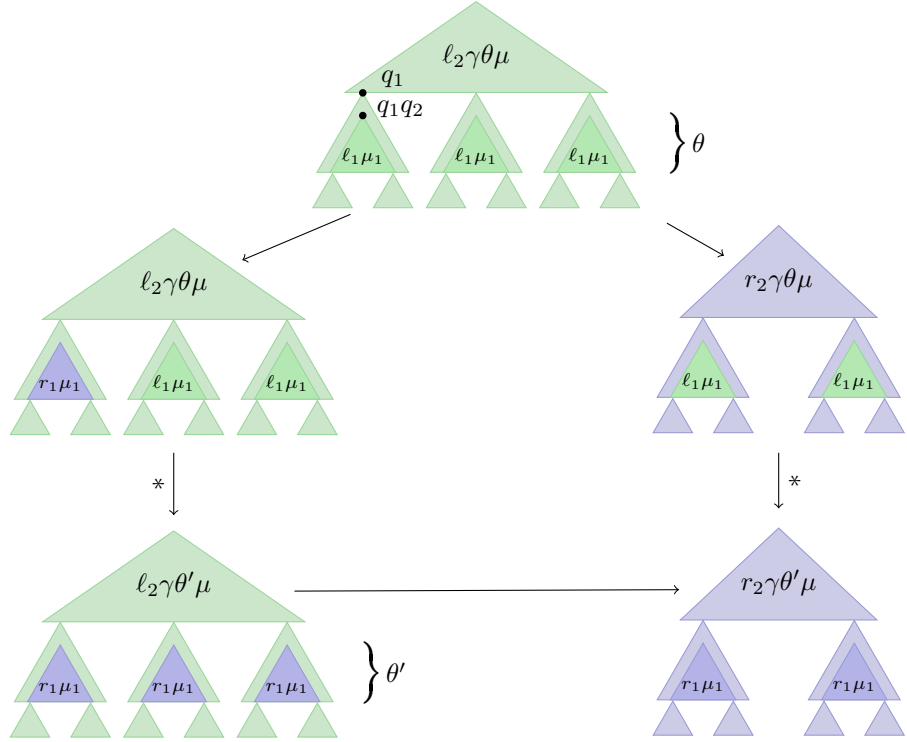
\begin{figure}[t]
\begin{center}
\begin{tikzpicture}[transform shape]
\begin{scope}
\node [l2,o'',yscale=.4] (1o) at (0,0) {};
\node [l2,m'] (1m1) at (1o.south) [xs=-15,ys=4.15] {};
\node [l2,m'] (1m2) at (1o.south) [ys=4.15] {};
\node [l2,m'] (1m3) at (1o.south) [xs=15,ys=4.15] {};
\node [l1,ic2] (1ic1) at (1m1.south) {};
\node [l1,ic2] (1ic2) at (1m2.south) {};
\node [l1,ic2] (1ic3) at (1m3.south) {};
\node [l2,ci'] (1ci11) at (1m1.south) [xs=-4,ys=5.1] {};
\node [l2,ci'] (1ci12) at (1m1.south) [xs=4,ys=5.1] {};
\node [l2,ci'] (1ci21) at (1m2.south) [xs=-4,ys=5.1] {};
\node [l2,ci'] (1ci22) at (1m2.south) [xs=4,ys=5.1] {};
\node [l2,ci'] (1ci31) at (1m3.south) [xs=-4,ys=5.1] {};
\node [l2,ci'] (1ci32) at (1m3.south) [xs=4,ys=5.1] {};
\node [as] at (1o) [above,ys=-2] {$\ell_2\gamma\theta\mu$};
\node at (1m1.north) [ys=-3.5,circle,is=1,fill] {};
\node at (1m1.north) [above,xs=4,ys=-4] {$q_1$};
\node at (1ic1.north) [ys=-3.5,circle,is=1,fill] {};
\node at (1ic1.north) [above,xs=5,ys=-5] {$q_1q_2$};
\node at (1ic1.south) [above,ys=1.5] {\scriptsize $\ell_1\mu_1$};
\node at (1ic2.south) [above,ys=1.5] {\scriptsize $\ell_1\mu_1$};
\node at (1ic3.south) [above,ys=1.5] {\scriptsize $\ell_1\mu_1$};
\node at (1o.south) [xs=28,ys=-6]
 {\rotatebox{-90}{$\overbrace{\makebox[8mm]{}}^%
 {\rotatebox{90}{$\theta$}}$}};
\end{scope}
\begin{scope}[xs=-40,ys=-30]
\node [l2,o'',yscale=.4] (3o) at (0,0) {};
\node [l2,m'] (3m1) at (3o.south) [xs=-15,ys=4.15] {};
\node [l2,m'] (3m2) at (3o.south) [ys=4.15] {};
\node [l2,m'] (3m3) at (3o.south) [xs=15,ys=4.15] {};
\node (3mm) at (3o.south) [ys=-15] {};
\node [r1,ic2] (3ic1) at (3m1.south) {};
\node [l1,ic2] (3ic2) at (3m2.south) {};
\node [l1,ic2] (3ic3) at (3m3.south) {};
\node [l2,ci'] (3ci11) at (3m1.south) [xs=-4,ys=5.1] {};
\node [l2,ci'] (3ci12) at (3m1.south) [xs=4,ys=5.1] {};
\node [l2,ci'] (3ci21) at (3m2.south) [xs=-4,ys=5.1] {};
\node [l2,ci'] (3ci22) at (3m2.south) [xs=4,ys=5.1] {};
\node [l2,ci'] (3ci31) at (3m3.south) [xs=-4,ys=5.1] {};
\node [l2,ci'] (3ci32) at (3m3.south) [xs=4,ys=5.1] {};
\node [as] at (3o) [above,ys=-2] {$\ell_2\gamma\theta\mu$};
\node at (3ic1.south) [above,ys=1.5] {\scriptsize $r_1\mu_1$};
\node at (3ic2.south) [above,ys=1.5] {\scriptsize $\ell_1\mu_1$};
\node at (3ic3.south) [above,ys=1.5] {\scriptsize $\ell_1\mu_1$};
\draw[->,sh=2] (1ci12) -- (3o);
\end{scope}
\begin{scope}[xs=-40,ys=-70]
\node [l2,o'',yscale=.4] (5o) at (0,0) {};
\node [l2,m'] (5m1) at (5o.south) [xs=-15,ys=4.15] {};
\node [l2,m'] (5m2) at (5o.south) [ys=4.15] {};
\node [l2,m'] (5m3) at (5o.south) [xs=15,ys=4.15] {};
\node [r1,ic2] (5ic1) at (5m1.south) {};
\node [r1,ic2] (5ic2) at (5m2.south) {};
\node [r1,ic2] (5ic3) at (5m3.south) {};
\node [l2,ci'] (5ci11) at (5m1.south) [xs=-4,ys=5.1] {};
\node [l2,ci'] (5ci12) at (5m1.south) [xs=4,ys=5.1] {};
\node [l2,ci'] (5ci21) at (5m2.south) [xs=-4,ys=5.1] {};
\node [l2,ci'] (5ci22) at (5m2.south) [xs=4,ys=5.1] {};
\node [l2,ci'] (5ci31) at (5m3.south) [xs=-4,ys=5.1] {};
\node [l2,ci'] (5ci32) at (5m3.south) [xs=4,ys=5.1] {};
\node [as] at (5o) [above,ys=-2] {$\ell_2\gamma\theta'\mu$};
\node at (5ic1.south) [above,ys=1.5] {\scriptsize $r_1\mu_1$};
\node at (5ic2.south) [above,ys=1.5] {\scriptsize $r_1\mu_1$};
\node at (5ic3.south) [above,ys=1.5] {\scriptsize $r_1\mu_1$};
\node at (5o.south) [xs=28,ys=-6]
 {\rotatebox{-90}{$\overbrace{\makebox[8mm]{}}^%
 {\rotatebox{90}{$\theta'$}}$}};
\draw[->,sh=2] (3mm) -- (5o) node [midway,below,sloped] {$*$};
\end{scope}
\begin{scope}[xs=40,ys=-30]
\node [r2,o',yscale=.54] (2o) at (0,0) {};
\node [r2,m'] (2ml) at (2o.south) [xs=-10,ys=4.4] {};
\node (2mm) at (2o.south) [ys=-15] {};
\node [r2,m'] (2mr) at (2o.south) [xs=10,ys=4.4] {};
\node [l1,ic2] (2icl) at (2ml.south) {};
\node [l1,ic2] (2icr) at (2mr.south) {};
\node [r2,ci'] (2cil1) at (2ml.south) [xs=-4,ys=5.1] {};
\node [r2,ci'] (2cil1) at (2ml.south) [xs=4,ys=5.1] {};
\node [r2,ci'] (2cir1) at (2mr.south) [xs=-4,ys=5.1] {};
\node [r2,ci'] (2cir2) at (2mr.south) [xs=4,ys=5.1] {};
\node [as] at (2o) [above,ys=-2] {$r_2\gamma\theta\mu$};
\node at (2icl.south) [above,ys=1.5] {\scriptsize $\ell_1\mu_1$};
\node at (2icr.south) [above,ys=1.5] {\scriptsize $\ell_1\mu_1$};
\draw[->,sh=2] (1ci32) -- (2o);
\end{scope}
\begin{scope}[xs=40,ys=-70]
\node [r2,o',yscale=.54] (4o) at (0,0) {};
\node [r2,m'] (4ml) at (4o.south) [xs=-10,ys=4.4] {};
\node [r2,m'] (4mr) at (4o.south) [xs=10,ys=4.4] {};
\node [r1,ic2] (4icl) at (4ml.south) {};
\node [r1,ic2] (4icr) at (4mr.south) {};
\node [r2,ci'] (4cil1) at (4ml.south) [xs=-4,ys=5.1] {};
\node [r2,ci'] (4cil1) at (4ml.south) [xs=4,ys=5.1] {};
\node [r2,ci'] (4cir1) at (4mr.south) [xs=-4,ys=5.1] {};
\node [r2,ci'] (4cir2) at (4mr.south) [xs=4,ys=5.1] {};
\node [as] at (4o) [above,ys=-2] {$r_2\gamma\theta'\mu$};
\node at (4icl.south) [above,ys=1.5] {\scriptsize $r_1\mu_1$};
\node at (4icr.south) [above,ys=1.5] {\scriptsize $r_1\mu_1$};
\draw[->,sh=2] (2mm) -- (4o) node[midway,above,sloped] {$*$};
\end{scope}
\draw[->,sh=7] (5o) -- (4o);
\end{tikzpicture}
\caption{Trivial variable overlap case in the proof of \lemref{cp}.}
\label{fig:cpCase4}
\end{center}
\end{figure}

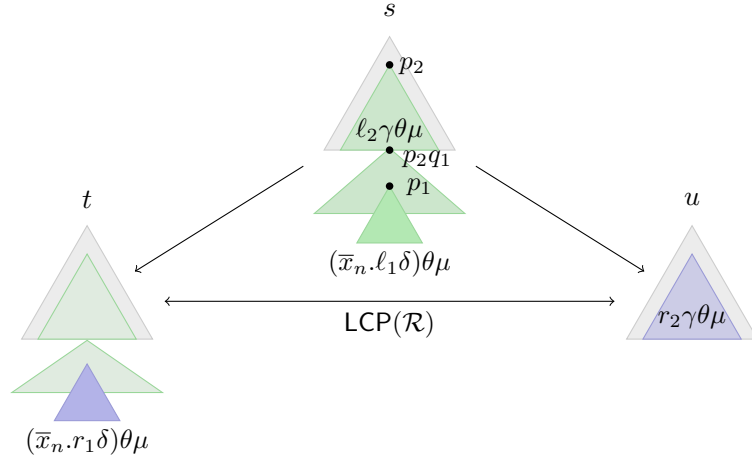
\begin{figure}[t]
\begin{center}
\begin{tikzpicture}[transform shape]
\begin{scope}
\node [c] (1) at (0,0) {};
\node [l2,o] (1o) at (0,0) {};
\node [l2,m,yscale=.5] (1m) at (1o.south) [ys=7.5] {};
\node [l1,i] (1i) at (1m.south) [ys=8] {};
\node at (1.north) {\large $s$};
\node at (1o.north) [ys=-3.5,circle,is=1,fill] {};
\node at (1o.north) [above,xs=3,ys=-6] {$p_2$};
\node [as] at (1o) [above,ys=-4] {$\ell_2\gamma\theta\mu$};
\node at (1m.north) [ys=-2,circle,is=1,fill] {};
\node at (1m.north) [above,xs=5,ys=-5.5] {$p_2q_1$};
\node at (1i.north) [ys=-3.5,circle,is=1,fill] {};
\node at (1i.north) [above,xs=4,ys=-6] {$p_1$};
\node at (1i.south) [below,ys=2]
 {$(\vec{x}.\ell_1\delta)\theta\mu$};
\end{scope}
\begin{scope}[xs=-40,ys=-25]
\node [c] (3) at (0,0) {};
\node [l2g,o] (3o) at (0,0) {};
\node [l2g,m,yscale=.4] (3m) at (3o.south) [ys=7.5] {};
\node [r1,i] (3i) at (3m.south) [ys=8] {};
\node at (3.north) {\large $t$};
\node at (3i.south) [below,ys=2] {$(\vec{x}.r_1\delta)\theta\mu$};
\draw[->,sh=2] (1) -- (3);
\end{scope}
\begin{scope}[xs=40,ys=-25]
\node [c] (2) at (0,0) {};
\node [r2,o] (2o) at (0,0) {};
\node at (2.north) {\large $u$};
\node at (2o.south) [above,ys=2] {$r_2\gamma\theta\mu$};
\draw[->,sh=2] (1) -- (2);
\end{scope}
\draw[<->,sh=7] (2) -- (3) node[midway,below] {$\lcp{\xR}$};
\end{tikzpicture}
\caption{Critical variable overlap case in the proof of \lemref{cp}.}
\label{fig:cpCase5}
\end{center}
\end{figure}

First assume $\theta(Y') = \vec[m]{y},\vec[l]{w}.\ell_1\mu_1$ for some
$\mu_1$. From
\[
\theta' = \SEt{Y' \mapsto \vec[m]{y},\vec[l]{w}.r_1\mu_1} \cup
\theta|_{\dom{\theta} \setminus \SET{Y'}}
\]
we obtain $\theta \Rab{\xR}{=} \theta'$ and therefore
$r_2\gamma\theta \Rab{\xR}{*} r_2\gamma\theta'$ by \thmref{clsubst}.
Without loss of generality, $\SEt{\vec{x},\vec[l]{w}}$ is fresh
for $\theta$. Furthermore, for
$\mu_2 = \SEt{\vec[m]{y} \mapsto \vec[m]{s}}$
we have $\vec{x},\vec[l]{w}.\ell_1\mu_1\mu_2 =
(\vec{x},\vec[l]{w}.Y'(\vec[m]{s},\vec[l]{w}\cf))\theta =
(\vec{x},\vec[l]{w}.\ell_1\delta)\theta =
\vec{x},\vec[l]{w}.\ell_1\delta\theta$,
so $\mu_1\mu_2|_{\fv{\ell_1}} = \delta\theta|_{\fv{\ell_1}}$. Hence,
\begin{align*}
(\vec{x},\vec[l]{w}.r_1\delta)\theta &=
\vec{x},\vec[l]{w}.r_1\delta\theta =
\vec{x},\vec[l]{w}.r_1\mu_1\mu_2 =
(\vec{x},\vec[l]{w}.Y'(\vec[m]{s}))\theta'
\end{align*}
and thus
\begin{align*}
t &= s[\vec[k]{z}.(\ell_2\gamma\theta[(\vec{x},\vec[l]{w}.
r_1\delta)\theta]_{q_11})\mu]_{p_2}
\Rab{\xR}{*} s[\vec[k]{z}.\ell_2\gamma\theta'\mu]_{p_2}
\Rb{\xR} s[\vec[k]{z}.r_2\gamma\theta'\mu]_{p_2} \\
&\Lab{\xR}{*} s[\vec[k]{z}.r_2\gamma\theta\mu]_{p_2} = u
\end{align*}
This case is illustrated in \figref{cpCase4}.

Otherwise, $\theta \in U$ and
$\langle \ell_1 \R r_1,q_1,1,(\vec{x}),(\vec[l]{w}),\delta,\gamma,U,
\ell_2 \R r_2 \rangle$ is an overlap and we obtain
\[
t = s[\vec[k]{z}.(\ell_2\gamma\theta[(\vec{x},\vec[l]{w}.
r_1\delta)\theta]_{q_11})\mu]_{p_2} \LRb{\lcp{\xR}}
s[\vec[k]{z}.r_2\gamma\theta\mu]_{p_2} = u
\]
as illustrated in \figref{cpCase5}.
\end{proof}

\section{Proof of \lemref{lcpjoin}}

In the following, we say that a substitution $\theta$ is
\emph{idempotent} if $\dom{\theta} \cap \fv{\theta} = \varnothing$.
Furthermore, given a permutation $\pi$ we write $\vec{x}^\pi$ for the
corresponding permutation of $\vec{x}$.

\begin{definition}
\label{def:lcprec}
Given a term $s = \vec[k]{x}.h(\vec{s})$, variables $Y$, $Y'$ of the
same type with arity $\geq 1$, a sequence of variables $\vec[l]{w}$ and a
substitution $\delta$,
we define $s^{Y,Y',\vec[l]{w},\delta}$ recursively as follows:
\begin{align*}
s^{Y,Y',\vec[l]{w},\delta} = \begin{cases}
\vec[k]{x}.Y'(\vec[l]{w}.s_1^{Y,Y',\vec[l]{w},\delta}\delta,
s_2^{Y,Y',\vec[l]{w},\delta},\dots,s_n^{Y,Y',\vec[l]{w},\delta})
&\text{if $h = Y$} \\
\vec[k]{x}.h(\vec{s^{Y,Y',\vec[l]{w},\delta}})
&\text{otherwise}
\end{cases}
\end{align*}
where $\SEt{\vec[k]{x}} $ is fresh for $\delta$. Given a substitution
$\theta$, $\theta^{Y,Y',\vec[l]{w},\delta} =
\SET{X \mapsto u^{Y,Y',\vec[l]{w},\delta} \mid X \mapsto u \in \theta}$.
\end{definition}

The preceding definition allows us to capture the connection
between ordinary and lifted critical pairs as stated in the upcoming
result.

\begin{lemma}
\label{lem:lcp_cp}
Let $\xR$ be a \textup{DPRS}. If $s \approx t \in \lcp{\xR}$ then
here exist $u \approx v \in \cp{\xR}$, variables $Y$, $Y'$ of the same
type with arity $\geq 1$, a sequence of variables $\vec[l]{w}$ and an
idempotent substitution $\delta$ with $\SET{Y,Y'}$ fresh for $\delta$
such that $s = u^{Y,Y',\vec[l]{w},\delta}\delta$ and
$t = v^{Y,Y',\vec[l]{w},\delta}\delta$.
\end{lemma}

\begin{proof}
Let $\langle \ell_1 \R r_1,p,q,(\vec{x}),(\vec[l]{w}),\delta_1,\gamma_1,U,
\ell_2 \R r_2 \rangle$ be the overlap corresponding to
$s \approx t \in \lcp{\xR}$, so
$s = \ell_2\gamma_1\theta_1[(\vec{x},\vec[l]{w}.r_1\delta_1)\theta_1]_{pq}$
and $t = r_2\gamma_1\theta_1$ for some $\theta_1 \in U$. If we are in case
\enuref{overlap-iv}(a) of \defref{overlap} then
$q = 0$, $l = 0$ and $\gamma_1 = \varnothing$. Hence, we can
choose fresh variables $Y$ and $Y'$ of the same type with arity $\geq 1$
as well as $\delta = \varnothing$ such that
$s = s^{Y,Y',(),\delta}\delta \approx t^{Y,Y',(),\delta}\delta = t
\in \cp{\xR}$ as desired.
Now consider case \enuref{overlap-iv}(b) of \defref{overlap} and
let $q = 1$, $\ell_2|_p = \vec{x}.Y(\vec[m]{s})$ and
$\gamma_1 = \SET{Y \mapsto \vec[m]{y}.Y_1''(\vec[l]{w}.
Y_1'(\vec[m]{y}\cf,\vec[l]{w}\cf),\vec[m]{y}\cf)}$.
If $l = 0$ then we proceed as in the previous case. So let $l > 0$ and
define $\gamma_2 = \SET{Y \mapsto \vec[m]{y}.Y_2''(Y_2'(\vec[m]{y}\cf),
\vec[m]{y}\cf)}$ and $\delta_2$ to be a $\vec{x}$-lifter of $\ell_1$ away
from $\fv{\ell_2\gamma_2}$. Note that
$\vec{x},\vec[l]{w}.\ell_1\delta_1$ and $\ell_2\gamma_1|_{p1}$ are
unifiable if and only if $\vec{x}.\ell_2\delta_2$ and
$\ell_2\gamma_2|_{p1}$ are; the two unification problems differ only
in the fact that in the first one, every free variable additionally has
$\vec[l]{w}\cf$ as arguments which do not occur elsewhere. Inspecting the
inference rules for DHP unification \cite{NM26a} reveals that this is an
invariant which holds at all times during the unification process.
Thus, there exists a unifier $\theta_2$ of $\vec{x}.\ell_2\delta_2$ and
$\ell_2\gamma_2|_{p1}$ which like $\theta_1$ does not satisfy one of the
conditions (1) and (2) in \defref{overlap}\enuref{overlap-iv}(b).
Hence, $u = \ell_2\gamma_2\theta_2[(\vec{x}.r_1\delta_2)\theta_2]_{pq}
\approx r_2\gamma_2\theta_2 = v \in \cp{\xR}$.
Moreover, there exists a variable renaming $\rho : \fv{\theta_2} \cup
\fv{\ell_2\delta_2} \to \fv{\theta_1} \cup \fv{\ell_1\delta_1}$
and permutation $\pi$ such that for
\[
\delta = \SET{X \mapsto \vec[m]{y}.\rho(X)((\vec[l]{w},\vec[m]{y})^\pi)
\mid X \in \fv{\theta_2} \cup \fv{\ell_2\delta_2}}
\]
we have $s = u^{Y_2'',Y_1'',\vec[l]{w},\delta}$ and
$t = v^{Y_2',Y_1',\vec[l]{w},\delta}$. Since $l > 0$,
$(\fv{\theta_2} \cup \fv{\ell_2\delta_2}) \cap
(\fv{\theta_1} \cup \fv{\ell_1\delta_1}) = \varnothing$,
so $\delta$ is idempotent. Furthermore, $\SET{Y_2'',Y_1''}$ is fresh for
$\delta$. Since $\dom{\delta}$ only occurs in the first arguments of
$Y_2''$, $s = s\delta$ and $t = t\delta$ as desired.
\end{proof}

We will use the following auxiliary result in order to show
that rewrite steps on (descendants of) terms in critical pairs
can also be performed in (descendants of) terms in lifted critical pairs.

\begin{lemma}
\label{lem:lcpsplit}
Consider a term $s$, variables $Y$, $Y'$ of the same type with arity
$\geq 1$, a sequence of variables $\vec[l]{w}$ and substitutions $\delta$,
$\theta$. If
\begin{enumerate}[(i)]
\item
$\delta$ is idempotent
\item
$\SET{Y,Y'}$ is fresh for $\delta$
\item
$\dom{\theta} \cap (\dom{\delta} \cup \SET{Y}) = \varnothing$
\end{enumerate}
then $s^{Y,Y',\vec[l]{w},\delta}\theta^{Y,Y',\vec[l]{w},\delta}\delta =
(s\theta)^{Y,Y',\vec[l]{w},\delta}\delta$.
\end{lemma}

\begin{proof}
We proceed by induction on $\ord{\theta}$ with a nested induction
on $s$ and abbreviate
$\cdot^{Y,Y',\vec[l]{w},\delta}$ to $\cdot^\bullet$.
Let $s = \vec[k]{x}.h(\vec{s})$. Without loss of generality,
$\SET{\vec[k]{x}}$ is fresh for $\theta$, $\theta^\bullet$ and $\delta$.
Consider the case where $h \notin \dom{\theta}$. If $h = Y$ then
\begin{align*}
s^\bullet\theta^\bullet\delta
&= \vec[k]{x}.Y'(\vec[l]{w}.s_1^\bullet\theta^\bullet\delta\delta,
s_2^\bullet\theta^\bullet\delta,\dots,s_n^\bullet\theta^\bullet\delta)
\tag{ii} \\
&= \vec[k]{x}.Y'(\vec[l]{w}.(s_1\theta)^\bullet\delta\delta,
(s_2\theta)^\bullet\delta,\dots,(s_n\theta)^\bullet\delta)
\tag{inner IH} \\
&= (s\theta)^\bullet\delta \tag{ii}
\end{align*}
and if $\delta(h) = \vec{y}.u$ then we have the following:
\begin{align*}
s^\bullet\theta^\bullet\delta
&= \vec[k]{x}.u\SET{\vec{y} \mapsto \vec{s^\bullet}\theta^\bullet\delta} \\
&= \vec[k]{x}.u\SET{\vec{y} \mapsto (\vec{s}\theta)^\bullet\delta}
\tag{inner IH} \\
&= \vec[k]{x}.u^\bullet\SET{\vec{y} \mapsto (\vec{s}\theta)^\bullet}\delta 
tag{i, ii} \\
&= \vec[k]{x}.(u\SET{\vec{y} \mapsto \vec{s}\theta})^\bullet\delta
\tag{outer IH} \\
&= (s\theta)^\bullet\delta
\end{align*}
Here, the outer IH is applicable because we can choose $\SEt{\vec{y}}$ to
satisfy condition (iii).
For the remaining case $h \notin \dom{\delta} \cup \SET{Y}$ we have
\begin{align*}
s^\bullet\theta^\bullet\delta
&= \vec[k]{x}.h(\vec{s^\bullet}\theta^\bullet\delta) \\
&= \vec[k]{x}.h((\vec{s}\theta)^\bullet\delta) \tag{inner IH} \\
&= (s\theta)^\bullet\delta
\end{align*}
as desired. Finally, if $\theta(h) = \vec{y}.u$ then
$h \notin \dom{\delta} \cup \SET{Y}$ by (iii) and
\begin{align*}
s^\bullet\theta^\bullet\delta
&= \vec[k]{x}.u^\bullet\delta\SET{\vec{y} \mapsto
\vec{s^\bullet}\theta^\bullet\delta}
\tag{wlog, $\SEt{\vec{y}}$ fresh for $\delta$} \\
&= \vec[k]{x}.u^\bullet\delta\SET{\vec{y} \mapsto
(\vec{s}\theta)^\bullet\delta} \tag{inner IH} \\
&= \vec[k]{x}.u^\bullet\SET{\vec{y} \mapsto
(\vec{s}\theta)^\bullet}\delta
\tag{wlog, $\SEt{\vec{y}}$ fresh for $\delta$} \\
&= \vec[k]{x}.(u\SET{\vec{y} \mapsto \vec{s}\theta})^\bullet\delta
\tag{outer IH} \\
&= (s\theta)^\bullet\delta
\end{align*}
Here, the outer IH is applicable because we can choose $\SEt{\vec{y}}$ to
satisfy condition (iii).
\end{proof}

\begin{corollary}
\label{cor:lcpsubst}
Consider a term $s$, variables $Y$, $Y'$ of the same type with arity
$\geq 1$, a sequence of variables $\vec[l]{w}$ and a substitutions
$\delta$, $\theta$. If
\begin{enumerate}[(i)]
\item
$\delta$ is idempotent
\item
$\SET{Y,Y'}$ is fresh for $\delta$
\item
$\dom{\theta} \cap (\dom{\delta} \cup \SET{Y}) = \varnothing$ and
\item
$Y \notin \fv{s}$
\end{enumerate}
then $s\theta^{Y,Y',\vec[l]{w},\delta}\delta =
(s\theta)^{Y,Y',\vec[l]{w},\delta}\delta$.
\end{corollary}

\begin{proof}
From (iv) we obtain $s = s^{Y,Y',\vec[l]{w},\delta}$, so
\lemref{lcpsplit} yields
$s\theta^{Y,Y',\vec[l]{w},\delta}\delta =
(s\theta)^{Y,Y',\vec[l]{w},\delta}\delta$.
\end{proof}

\begin{lemma}
\label{lem:lcpstep}
Consider a \textup{DPRS}, terms $s$, $t$
variables $Y$, $Y'$ of the same type with arity $\geq 1$,
a sequence of variables $\vec[l]{w}$ and an idempotent
substitution $\delta$ with $\SET{Y,Y'}$ fresh for $\delta$.
If $s \RbR t$ then $s^{Y,Y',\vec[l]{w},\delta}\delta \Rab{\xR}{*}
t^{Y,Y',\vec[l]{w},\delta}\delta$.
\end{lemma}

\begin{proof}
We proceed by induction on $s$. Let $s = \vec[k]{x}.h(\vec{s})$.
If $s \RbR s[u]_i = t$ then
$s_i^{Y,Y',\vec[l]{w},\delta}\delta \Rab{\xR}{*}
u^{Y,Y',\vec[l]{w},\delta}\delta$ by
the induction hypothesis. If $h \neq Y$ or $i > 1$ we immediately obtain
$s^{Y,Y',\vec[l]{w},\delta}\delta \Rab{\xR}{*}
t^{Y,Y',\vec[l]{w},\delta}\delta$. For $h = Y$ and $i = 1$ we obtain
$\vec[l]{w}.s_1^{Y,Y',\vec[l]{w},\delta}\delta \Rab{\xR}{*}
\vec[l]{w}.u^{Y,Y',\vec[l]{w},\delta}\delta$ and therefore also
$s^{Y,Y',\vec[l]{w},\delta}\delta \Rab{\xR}{*}
t^{Y,Y',\vec[l]{w},\delta}\delta$ as $\delta$ is idempotent.
Now consider $s \Rab{\xR}{\epsilon} t$. By definition,
$s = \vec[k]{x}.\ell\theta$ for some rule $\ell \R r \in \xR$
and substitution $\theta$ with $\dom{\theta} \subseteq \fv{\ell}$.
Without loss of generality, $\fv{\ell} \supseteq \fv{r}$ is fresh. Hence,
\corref{lcpsubst} yields $u\theta^{Y,Y',\vec[l]{w},\delta}\delta =
(u\theta)^{Y,Y',\vec[l]{w},\delta}\delta$ for $u \in \SET{\ell,r}$
and thus $s^{Y,Y',\vec[l]{w},\delta}\delta \RbR
t^{Y,Y',\vec[l]{w},\delta}\delta$.
\end{proof}

\lcpjoin*

\begin{proof}
Let $\xR$ be a DPRS and $s \approx t \in \lcp{\xR}$.
By \lemref{lcp_cp} there exist $u \approx v \in \cp{\xR}$, variables
$Y$, $Y'$ of the same type with arity $\geq 1$, a sequence of variables
$\vec[l]{w}$ and an idempotent substitution $\delta$ with $\SET{Y,Y'}$
fresh for $\delta$ such that $s = u^{Y,Y',\vec[l]{w},\delta}\delta$ and
$t = v^{Y,Y',\vec[l]{w},\delta}\delta$.
By assumption $u \Db{\xR} v$, so from \lemref{lcpstep}
together with a straightforward induction argument
we obtain $s \Db{\xR} t$ as desired.
\end{proof}

\section{Remaining Proofs in \secref{ortho}}
\label{app:proofs}

We start with two easy but important results. The first
one follows by a straightforward induction on terms by using case
\enuref{ms1} of the multi-step relation. The second one is an easy
consequence of the fact that rules consist of basic terms.

\begin{lemma}
\label{lem:msRefl}
For every \textup{DPRS} $\xR$, the relation $\MSb{\xR}$ is reflexive.
\end{lemma}

\begin{lemma}
\label{lem:msAbs}
We have $\vec[k]{x}.s \MSb{\xR} \vec[k]{x}.t$ if and only if
$s \MSb{\xR} t$.
\end{lemma}

\msincl*

\begin{proof}
Consider $s \Rab{\xR}{p} t$. We proceed by induction on $p$. If
$p = \epsilon$ then $s \MSb{\xR} t \mscase{ii}$ using the substitution
$\theta$ from the rewrite step since we obtain $\theta \MSb{\xR} \theta$
from \lemref{msRefl}. Now let $p = iq$, so $s = \vec[k]{x}.h(\vec{s})$ and
$t = \vec[k]{x}.h(s_1,\dots,t_i,\dots,s_n)$ for $n \geq i$.
From \lemref{msRefl} we obtain $\vec[k]{x}.s_j \MSb{\xR} \vec[k]{x}.s_j$
for all $1 \leq j \leq n$ and the induction hypothesis yields
$\vec[k]{x}.s_i \MSb{\xR} \vec[k]{x}.t_i$. Hence,
$s \MSb{\xR} t \mscase{i}$. For the second inclusion, consider
$s \MSb{\xR} t$. We proceed by induction on the derivation of
$s \MSb{\xR} t$. If $s \MSb{\xR} t\mscase{i}$ then
$s = \vec[k]{x}.h(\vec{s})$, $t = \vec[k]{x}.h(\vec{t})$ and the induction
hypothesis yields $\vec[k]{x}.s_i \Rab{\xR}{*} \vec[k]{x}.t_i$. Hence, we
obtain $s \Rab{\xR}{*} t$ by sequentialization. Now assume
$s \MSb{\xR} t \mscase{ii}$, so $s = \vec[k]{x}.\ell\theta$,
$t = \vec[k]{x}.r\delta$ and we obtain $\theta(x) \Rab{\xR}{*} \delta(x)$
for all $x \in \fv{\ell}$ from the induction hypothesis. Hence,
$s \Rab{\xR}{*} t$ by \thmref{clsubst} together with \lemref{msAbs}.
\end{proof}

Next, we establish a more general form of closure under substitutions
for multi-steps. Note that proving closure under substitutions directly
does not work \cite{MN98}.

\begin{lemma}
\label{lem:msSubst}
If $\theta \MSb{\xR} \delta$ and $s \MSb{\xR} t$ then
$s\theta \MSb{\xR} t\delta$.
\end{lemma}

\begin{proof}
Without loss of generality
$s = \vec[k]{x}.h(\vec{s})$ such that
$\SEt{\vec[k]{x}}$ is fresh for $\theta$ and $\delta$.
We proceed by induction on $\ord{\theta}$ with a nested
induction on the derivation of $s \MSb{\xR} t$.
If $s \MSb{\xR} t \mscase{i}$ then
$t = \vec[k]{x}.h(\vec{t})$ and $\vec[k]{x}.s_i \MSb{\xR} \vec[k]{x}.t_i$
for all $1 \leq i \leq n$. The inner induction hypothesis yields
$\vec[k]{x}.s_i\theta \MSb{\xR} \vec[k]{x}.t_i\delta$ for all
$1 \leq i \leq n$ (\textasteriskcentered). We distinguish two cases.
\begin{itemize}
\item
If $h \notin \dom{\theta}$ then also $h \notin \dom{\delta}$ by our
assumption $\theta \MSb{\xR} \delta$. Hence,
$s\theta = \vec[k]{x}.h(\vec{s}\theta) \MSb{\xR}
\vec[k]{x}.h(\vec{t}\delta) = t\delta$ as desired.
\item
Let $\theta(h) = \vec{y}.u$ and $\delta(h) = \vec{y}.v$. By assumption,
$\vec{y}.u \MSb{\xR} \vec{y}.v$ and therefore also $u \MSb{\xR} v$ by
\lemref{msAbs}. Define $\gamma = \SET{\vec{y} \mapsto \vec{s}\theta}$ and
$\mu = \SET{\vec{y} \mapsto \vec{t}\delta}$. Note that
$\ord{\theta} > \ord{\gamma}$. Furthermore, from (\textasteriskcentered) we
obtain $\gamma \MSb{\xR} \mu$. Hence, $u\gamma \MSb{\xR} v\mu$ by the outer
induction hypothesis. Thus, also
$s\theta  = \vec[k]{x}.u\gamma \MSb{\xR} \vec[k]{x}.v\mu = t\delta$ by
\lemref{msAbs}.
\end{itemize}
Now consider the case where $s \MSb{\xR} t\mscase{ii}$. We have
$s = \vec[k]{x}.\ell\gamma$ and $t = \vec[k]{x}.r\mu$ for some rule
$\ell \R r \in \xR$. By assumption, $\gamma \MSb{\xR} \mu$. From the inner
induction hypothesis we obtain $\gamma\theta \MSb{\xR} \mu\delta$.
Hence, from $s\theta = \vec[k]{x}.\ell\gamma\theta$ and
$t\delta = \vec[k]{x}.r\mu\delta$ we obtain $s\theta \MSb{\xR} t\delta$ as
desired.
\end{proof}

In the following, we say that an
element $a \in A$ is \emph{terminal} in a binary relation
$R \subseteq A \times A$ if $a \mathrel{R} b$ implies $b = a$.
Next, we show coherence of terminal anchors in redexes:
A step with $\MS$ starting from a term containing terminal anchors
can also be performed when one abstracts away from these terminal anchors.

\begin{lemma}
\label{lem:lhsSubtAux}
Let $\xR$ be an orthogonal \textup{DPRS},
$\ell \subt \vec[k]{x}.F(\vec{s})$
for some $\ell \R r \in \xR$ and let $\vec[k]{x}.s_i$ be terminal in
$\MSb{\xR}$ for all $1 \leq i \leq n$.
If $\vec[k]{x}.u\SEt{\vec{y} \mapsto \vec{s}} \MSb{\xR} \vec[k]{x}.t$
then there is a term $\vec{y}.v$ such that $\vec{y}.u \MSb{\xR} \vec{y}.v$
and $\vec[k]{x}.t = \vec[k]{x}.v\SEt{\vec{y} \mapsto \vec{s}}$.
\end{lemma}

\begin{proof}
Let $\mu = \SEt{\vec{y} \mapsto \vec{s}}$. We proceed by induction on
$u = \vec[l]{z}.h_u(\vec[m]{u})$. Without loss of generality,
$\SEt{\vec[l]{z}}$ is fresh for $\mu$. Assume 
$\vec[k]{x}.u\mu \MSb{\xR} \vec[k]{x}.t$. We distinguish two cases.
\begin{itemize}
\item
Let $\vec[k]{x}.u\mu \MSb{\xR} \vec[k]{x}.t \mscase{i}$. If
$h_u \in \xF \cup (\xV \setminus \SEt{\vec{y}})$ then
$\vec[k]{x}.u\mu = \vec[k]{x},\vec[l]{z}.h_u(\vec[m]{u}\mu)$
and therefore $\vec[k]{x}.t = \vec[k]{x},\vec[l]{z}.h_u(\vec[m]{t})$ as
well as $\vec[k]{x},\vec[l]{z}.u_i\mu \MSb{\xR} \vec[k]{x},\vec[l]{z}.t_i$
for all $1 \leq i \leq m$. From the induction hypothesis, we obtain
the existence of terms $\vec{y},\vec[l]{z}.v_i$ such that
$\vec[k]{x},\vec[l]{z}.u_i \MSb{\xR} \vec[k]{x},\vec[l]{z}.v_i$ and
$\vec[k]{x},\vec[l]{z}.t_i = \vec[k]{x},\vec[l]{z}.v_i\mu$ for all
$1 \leq i \leq m$. Hence, for the term
$\vec{y}.v = \vec{y},\vec[l]{z}.h_u(\vec[m]{v})$ we have
$\vec{y}.u \MSb{\xR} \vec{y}.v$ and $\vec[k]{x}.t = \vec[k]{x}.v\mu$ as
desired. Otherwise, $h_u = y_j$ for some $1 \leq j \leq n$. Let
$\vec[k]{x}.s_j = \vec[k]{x},\vec[m]{z}.h_j(\vec[o]{s_j},\vec[m]{z}\cf)$,
so
$\vec[k]{x}.u\mu = \vec[k]{x},\vec[l]{z}.h_j(\vec[o]{s_j},\vec[m]{u}\mu)$.
Hence, $\vec[k]{x}.t = \vec[k]{x},\vec[l]{z}.h_j(\vec[o+m]{t})$ and
$\vec[k]{x},\vec[l]{z}.{s_j}_i \MSb{\xR} \vec[k]{x},\vec[l]{z}.t_i$ for
$1 \leq i \leq o$ as well as
$\vec[k]{x},\vec[l]{z}.u_i\mu \MSb{\xR} \vec[k]{x},\vec[l]{z}.t_{o+i}$ for
$1 \leq i \leq m$. By assumption, the term $\vec[k]{x}.s_j$ is
terminal in
$\MSb{\xR}$, so $\vec[x]{x},\vec[l]{z}.t_i = \vec[k]{x},\vec[l]{z}.{s_j}_i$
for all $1 \leq i \leq o$. Furthermore, the induction hypothesis yields
terms $\vec{y},\vec[l]{z}.v_i$ such that
$\vec[k]{x},\vec[l]{z}.u_i \MSb{\xR} \vec[k]{x},\vec[l]{z}.v_i$ and
$\vec[k]{x},\vec[l]{z}.t_{o+i} = \vec[k]{x},\vec[l]{z}.v_i\mu$ for all
$1 \leq i \leq m$. Hence, for the term
$\vec{y}.v = \vec{y},\vec[l]{z}.h_u(\vec[m]{v})$ we have
$\vec{y}.u \MSb{\xR} \vec{y}.v$ and $\vec[k]{x}.t = \vec[k]{x}.v\mu$ as
desired.
\item \smallskip
Now let $\vec[k]{x}.u\mu \MSb{\xR} \vec[k]{x}.t \mscase{ii}$, so
$\vec[k]{x}.u\mu = \vec[k]{x},\vec[l]{z}.\ell'\gamma_1$ and
$\vec[k]{x}.t = \vec[k]{x},\vec[l]{z}.r\gamma_2$ for some
$\ell' \R r' \in \xR$ and substitutions $\gamma_1$, $\gamma_2$ such that
$\dom{\gamma_1} \subseteq \fv{\ell'}$ and $\gamma_1 \MSb{\xR} \gamma_2$.
Let $\vec{y}.u = \vec{y},\vec[l]{z}.u'$ and $F'$ by a fresh variable such
that $F'$ and $\vec{y}.u'$ have the same type. Then,
$(\vec[k]{x}.F'(\vec{s}))\SET{F' \mapsto \vec{y}.u'} =
\vec[k]{x}.\ell'\gamma_1$, so by \lemref{unifierExistence},
$\vec[k]{x}.\ell'\delta$ and $\vec[k]{x}.F'(\vec{s})$ are unifiable by
some substitution $\theta_0 \supseteq \SET{F' \mapsto \vec{y}.u'}$ where
$\delta$ is a $\vec[k]{x}$-lifter of $\ell'$ away from $\fv{\ell}$.
Let $U$ be a minimal complete set of unifiers of the two terms, so there
exists $\theta \in U$ such that $\theta \subsumes \theta_0$.
If $(\vec[k]{x}.F'(\vec{s}))\theta = \vec[k]{x}.s_j$ for some
$1 \leq j \leq n$ then $\vec{y}.u' = \vec{y}.y_j$. Furthermore,
$\vec[k]{x}.u\mu = \vec[k]{x},\vec[l]{z}.s_j = \vec[k]{x}.t$ as
$\vec[k]{x}.s_j$ is terminal in $\MSb{\xR}$. Hence, we can choose
$\vec{y}.v = \vec{y}.u$. Otherwise, since $\xR$ is non-ambiguous, we have
$\vec{y}.\ell' \subsumes \theta(F') \subsumes \theta_0(F')$. Together with
$\theta_0 \supseteq \SET{F' \mapsto \vec{y}.u'}$ and
$\vec[k]{x}.u'\mu = \vec[k]{x}.\ell\gamma_1$ we get
$\vec{y}.u' = \vec{y}.\ell'\gamma_1'$ where $\gamma_1 = \gamma_1'\mu$.
Furthermore, we also obtain the existence of $\gamma_2'$ such that
$\delta = \gamma_2'\mu$ from $\gamma_1 \MSb{\xR} \gamma_2$ and our
assumption that $\vec[k]{x}.s_i$ is terminal in $\MSb{\xR}$ for all
$1 \leq i \leq n$. Hence, we have $\gamma_1' \MSb{\xR} \gamma_2'$ and
obtain $\vec{y}.u = \vec{y},\vec[l]{z}.\ell\gamma_1' \MSb{\xR}
\vec{y},\vec[l]{z}.r\gamma_2'$ as well as
$\vec[k]{x},\vec[l]{z}.r\gamma_2'\mu = \vec[k]{x},\vec[l]{z}.r\gamma_2 =
\vec[k]{x}.t$ as desired. \qedhere
\end{itemize}
\end{proof}

Using \lemref{lhsSubtAux},
we can now prove the general case of coherence for orthogonal DPRSs.
We will obtain the assumption that anchors are terminal for free
from the induction hypothesis.

\lhssubt*

\begin{proof}
By assumption, $s = \vec[k]{x}.h(\vec{s})$ is a linear term. Furthermore,
$\SEt{\vec[k]{x}}$ is fresh for $\theta$ without loss of generality.
We perform induction on $s$ and distinguish two cases.
\begin{itemize}
\item
Assume $h \notin \dom{\theta}$. If $s\theta \MSb{\xR} t \mscase{i}$ then
$t = \vec[k]{x}.h(\vec{t})$ and
$(\vec[k]{x}.s_i)\theta \MSb{\xR} \vec[k]{x}.t_i$ for all
$1 \leq i \leq n$.
By the induction hypothesis, there exist substitutions $\vec{\delta}$ such
that $\vec[k]{x}.t_i = (\vec[k]{x}.s_i)\delta_i$ and
$\theta|_{\fv{\vec[k]{x}.s_i}} \MSb{\xR} \delta_i$ for all
$1 \leq i \leq n$. Without loss of generality, $\SEt{\vec[k]{x}}$ is fresh
for $\delta_i$ for all $1 \leq i \leq n$. Since $s$ is linear, the
substitution $\delta = \bigcup_{i = 1}^n \delta_i$ is well-defined and
satisfies $\theta|_{\fv{s}} \MSb{\xR} \delta$. Finally,
$s\delta = \vec[k]{x}.h(\vec{s}\delta) = \vec[k]{x}.h(\vec{t}) = t$.
Now suppose $s\theta \MSb{\xR} t\mscase{ii}$, so there exist a rewrite rule
$\ell' \R r' \in \xR$ and a substitution $\theta'$ such that
$s\theta = \vec[k]{x}.\ell'\theta'$. In particular,
$h \in \xF$. Let $p$ be the position such that $\ell|_p = s$.
By \lemref{unifierExistence}, $\vec{x}.\ell'\mu$ and $s$ are unifiable
where
$\mu$ is a $\vec[k]{x}$-lifter of $\ell'$ away from $\fv{\ell}$. Let $U$ be
a complete set of unifiers of $\vec{x}.\ell'\mu$ and $s$. Then,
$\langle \ell' \R r', p, \epsilon, \vec[k]{x}, (), \mu, \varnothing,
U, \ell \R r \rangle$ is an overlap in $\xR$ which contradicts the fact
that orthogonal DPRSs are non-ambiguous.
\item \smallskip
Let $\theta(h) = \vec{y}.u = \vec{y}.h_u(\vec[m]{u})$, so
$s\theta = \vec[k]{x}.u\SEt{\vec{y} \mapsto \vec{s}}$. In particular,
$\fv{s} = \SET{h}$ as $\ell$ and therefore $s$ is a DHP. Thus, the
induction hypothesis yields that $\vec[k]{x}.s_i$ is terminal in
$\MSb{\xR}$ for all $1 \leq i \leq n$. From \lemref{lhsSubtAux}, we
obtain a term $\vec{y}.v$ such that $\vec{y}.u \MSb{\xR} \vec{y}.v$ and
$t = \vec[k]{x}.v\SEt{\vec{y} \mapsto \vec{s}}$. Hence, we can define
$\delta = \SET{h \mapsto \vec{y}.v}$ to obtain
$s\delta = \vec[k]{x}.v\SEt{\vec{y} \mapsto \vec{s}} = t$ and
$\theta(h) \MSb{\xR} \delta(h)$ as desired.
\end{itemize}
\end{proof}

\multisteptriangle*

\begin{proof}
We perform induction on the derivation of $s \MSb{\xR} t$. If
$s \MSb{\xR} t\mscase{i}$ then $s = \vec[k]{x}.h(\vec{s})$,
$t = \vec[k]{x}.h(\vec{t})$ and $\vec[k]{x}.s_i \MSb{\xR} \vec[k]{x}.t_i$
for all $1 \leq i \leq n$. We proceed by case analysis on whether
$h(\vec{s})$ is a redex.
\begin{itemize}
\item
If $h(\vec{s})$ is not a redex then $u = \vec[k]{x}.h(\vec{u})$ and
$\vec[k]{x}.s_i \MMSb{\xR} \vec[k]{x}.u_i$ for all $1 \leq i \leq n$.
From the induction hypothesis we obtain
$\vec[k]{x}.t_i \MSb{\xR} \vec[k]{x}.u_i$
and therefore $t \MSb{\xR} u$ as desired.
\item \smallskip
If $h(\vec{s})$ is a redex then
$\vec[k]{x}.h(\vec{s}) = \vec[k]{x}.\ell\theta$ for some rule
$\ell \R r \in \xR$ and substitution $\theta$ with
$\dom{\theta} \subseteq \fv{\ell}$. By definition,
$h \in \xF$ and $\ell = h(\vec{\ell})$. Fix $i \in \SET{1,\dots,n}$. By
assumption,
$\vec[k]{x}.s_i = \vec[k]{x}.\ell_i\theta \MSb{\xR} \vec[k]{x}.t_i$. From
\lemref{lhsSubt} we obtain substitutions $\delta_i$ such that
$\vec[k]{x}.t_i = \vec[k]{x}.\ell_i\delta_i$ and
$\theta|_{\fv{\ell_i}} \MSb{\xR} \delta_i$ for $1 \leq i \leq n$. Without
loss of generality, we may assume that
$\dom{\delta_i} \subseteq \fv{\ell_i}$. Since $\ell$ is linear,
$\delta = \bigcup_{i=1}^{n}\delta_i$ is well-defined. Furthermore,
$t = \vec[k]{x}.\ell\delta$ and $\theta \MSb{\xR} \delta$. By assumption,
$u = \vec[k]{x}.r\gamma$ such that $\theta \MMSb{\xR} \gamma$. From
$|\fv{\ell}|$ applications of the induction hypothesis we obtain
$\delta \MSb{\xR} \gamma$ and therefore
$t = \vec[k]{x}.\ell\delta \MSb{\xR} \vec[k]{x}.r\gamma = u$ as desired.
\end{itemize}
Now consider $s \MSb{\xR} t\mscase{ii}$, so $s = \vec[k]{x}.\ell\theta$,
$t = \vec[k]{x}.r\delta$ and $\theta \MSb{\xR} \delta$.
By assumption, $u = \vec[k]{x}.r\gamma$ such that
$\theta \MMSb{\xR} \gamma$.
As before, by applying the induction hypothesis $|\fv{\ell}|$ times
we obtain $\delta \MSb{\xR} \gamma$ and therefore
$t = \vec[k]{x}.r\delta \MSb{\xR} \vec[k]{x}.r\gamma = u$ by
Lemmata~\ref{lem:msRefl}, \ref{lem:msAbs} and \ref{lem:msSubst}.
\end{proof}

\end{document}